\documentclass[reqno]{amsart}

\usepackage[T1]{fontenc}
\usepackage[utf8]{inputenc}

\usepackage{amsmath, amsfonts}
\usepackage{amsthm}
\usepackage{amssymb}
\usepackage{dsfont}
\usepackage{amscd}
\usepackage{color}
\usepackage{graphicx}

\usepackage{hyperref}
\definecolor{darkred}{rgb}{0.8,0,0}

\newtheorem{theorem}{Theorem}[section]
\newtheorem{corollary}[theorem]{Corollary}
\newtheorem{lemma}[theorem]{Lemma}
\newtheorem{prop}[theorem]{Proposition}
\theoremstyle{definition}

\theoremstyle{remark}
\newtheorem{remark}[theorem]{Remark}

\numberwithin{equation}{section}

\let\Re\undefined
\DeclareMathOperator{\Re}{\rm Re}

\DeclareMathOperator{\II}{\mathbb I}
\DeclareMathOperator{\KK}{\mathbb K}
\DeclareMathOperator{\LL}{\mathbb L}
\DeclareMathOperator{\TT}{\mathbb T}
\DeclareMathOperator{\tr}{tr}
\DeclareMathOperator{\omeg}{\Omega}
\let\O\undefined
\DeclareMathOperator{\O}{O}
\newcommand{\Om}[2]{\ensuremath{\omeg(#1):#2}}

\usepackage{listings}
\definecolor{mygreen}{rgb}{0,0.6,0}
\definecolor{mygray}{rgb}{0.5,0.5,0.5}
\definecolor{lightgray}{rgb}{0.925,0.925,0.925}
\definecolor{mymauve}{rgb}{0.58,0,0.82}
\definecolor{gcolor}{rgb}{0.004,0.396,0.741}
\begin{document}

\title[Finite size effects for spacing distributions]{Finite size effects for spacing distributions in random matrix theory: circular ensembles and Riemann zeros}

\author{Folkmar Bornemann}
\address[Folkmar Bornemann]{Zentrum Mathematik -- M3, 
  Technische Universit\"at M\"unchen, Germany}
\email{bornemann@tum.de}

\author{Peter J. Forrester} \address[Peter J. Forrester]{Department of Mathematics and Statistics, 
ARC Centre of Excellence for Mathematical \& Statistical Frontiers,
The University of Melbourne, Victoria 3010, Australia
}\email{pjforr@unimelb.edu.au}

\author{Anthony Mays} \address[Anthony Mays]{Department of Mathematics and Statistics, 
ARC Centre of Excellence for Mathematical \& Statistical Frontiers,
The University of Melbourne, Victoria 3010, Australia
}\email{anthony.mays@unimelb.edu.au}

\date{}

\begin{abstract} According to Dyson's three fold way,
from the viewpoint of global time reversal symmetry
 there are three circular ensembles of unitary random matrices relevant to the study of 
chaotic spectra in quantum mechanics.
These are the circular orthogonal, unitary and symplectic ensembles, denoted
COE, CUE and CSE respectively.
For each of these three ensembles and their thinned versions, whereby each eigenvalue is deleted independently with probability $1-\xi$,
we take up the problem of calculating the first two terms in the scaled large $N$ expansion of the spacing distributions. It is well
known that the leading term admits a characterisation in terms of both Fredholm determinants and Painlev\'e transcendents.
We show that modifications of these characterisations also remain valid for the next to leading term, and that they provide schemes
for high precision numerical computations. In the case of the CUE there is an application to the analysis of Odlyzko's data set for the
Riemann zeros, and in that case some further statistics are similarly analysed.
\end{abstract}

\maketitle
\section{Introduction}
Random matrix theory is of importance for both its conceptual and its predictive powers. At a conceptual level, there is the three fold way 
classification of global time reversal symmetries of Hamiltonians  for chaotic quantum systems \cite{Dy62c}, later
extended to ten by the inclusion of chiral and topological symmetries \cite{AZ97}. For each of the ten Hamiltonians there is an ensemble
of random Hermitian matrices, which in turn are the Hermitian part of the ten matrix Lie algebras associated with the infinite families of symmetric
spaces. The symmetric spaces themselves give rise to ten  ensembles of random
unitary matrices \cite{KS99}.

To see how the formulation of ensembles leads to a predictive statement, consider for definiteness the three fold way classification of 
Hamiltonians  for chaotic quantum systems. The basic hypothesis, initiated by Wigner \cite{Wi57} and Dyson  \cite{Dy62c}, and refined by Bohigas et al.~\cite{BGS84},
is that a quantum mechanical system for which the underlying classical mechanics is chaotic, will, for large energies have the same
statistical distribution as the bulk scaled eigenvalues of the relevant ensemble of Hermitian matrices.  The particular ensemble is 
determined by the
presence or absence of a time reversal symmetry, with the former consisting of two subcases depending on the time reversal
operator $T$ being such that $T^2 = 1$, or $T^2 = -1$; see e.g.~\cite[Ch.~1]{Fo10}. This becomes predictive upon the specification of
the statistical distributions for the bulk scaled eigenvalues of the random matrix
ensemble. Thus, according to the hypothesis, these same distributions will
be exhibited by the point process formed from the 
highly excited energy levels of a chaotic quantum system.
The point process can be realised in laboratory experiments of nuclear excitations (see e.g.~\cite{FM09a}), and microwave billiards \cite{St10}, amongst other examples. 

Arguably the most spectacular example of both the conceptual and predictive powers of random matrix theory shows itself in relation
to the Riemann zeta function $\zeta(s)$. The celebrated Riemann hypothesis asserts that all the complex zeros of $\zeta(s)$ are
of the form $s = {1 \over 2} \pm i E$, with $E> 0$. The Montgomery-Odlyzko law asserts that for large $E$ these zeros --- referred to as the
Riemann zeros --- have the same statistical properties as the bulk eigenvalues of large Hermitian random matrices with complex
elements. But from the three fold way, the latter is the ensemble of random Hermitian matrices giving the statistical properties of the
large energy levels of a chaotic quantum mechanical system without a time reversal symmetry. The implied relationship between the Riemann zeros
and quantum chaos is consistent with and extends the Hilbert-P\'olya conjecture \cite{WikiHP} asserting that the Riemann zeros
correspond to the eigenvalues of some (unknown) unbounded self-adjoint operator; see \cite{BK99,SR11} for contemporary research on this topic.

To realise the predictive consequences of this link between two seemingly unrelated quantities --- the Riemann zeros and
chaotic quantum Hamiltonians without time reversal symmetry --- requires a list of the Riemann zeros for large $E$.
Such a list has been provided by Odlyzko \cite{Od87}, who in a celebrated numerical computation has provided the
high precision evaluation of the $10^{20}$-th Riemann zero and over 70 million of its neighbours. 
From this data set the veracity of the Montgomery-Odlyzko law can be tested. Primary statistical quantities for this purpose include
the two-point correlation function $\langle \rho_{(2)}(x,x+s) \rangle$ --- which measures the density of zeros a distance $s$ from $x$,
with $x$ itself averaged over a window of zeros which itself is large, but still small relative to $x$ --- and the spacing distribution $p(k;s)$ for
the event that the $k$-th next (in consecutive order) neighbours\footnote{In an increasingly ordered list the
$k$-th next neighbour of the item $x_j$ is item $x_{j+k+1}$.} are a distance $s$ apart. As displayed in \cite{Od87}, at a graphical level the agreement
between the Riemann zero data and the predictions of random matrix theory is seemingly exact.

Although zeros of order $10^{20}$ are huge on an absolute scale, it turns out that convergence to random matrix predictions
occurs on a scale where $\log E$ plays the role of $N$, giving rise to potentially significant finite size effects.
Such effects were first considered in 
the work of Keating and Snaith  \cite{KS00a} in their study of the statistical properties of the value distribution of the logarithm of the Riemann zeta
function on $\Re s = 1/2$ with $|s|$ large. Quite unexpectedly they showed that the finite size corrections could be understood by
introducing as a model for the Riemann zeros the eigenvalues of Haar distributed unitary random matrices from U$(N)$, instead of complex Hermitian random matrices.
It had been known since the work of Dyson \cite{Dy62e} that the eigenvalues for Haar distributed unitary random matrices from U$(N)$ and
complex Hermitian random matrices with Gaussian entries have the same limiting statistical distribution. But for finite $N$ they
are fundamentally different, with the former being rotationally invariant while the spectral density of the latter is given by the
Wigner semi-circle law and thus has edge effects.

Odlyzko's work on the computation of large Riemann zeros is ongoing. In \cite{Od01} he announced a data set beginning with
the $10^{23}$-rd zero and its next $10^9$ neighbours. Taking advantage of the great statistical accuracy offered by this data set,
the finite size correction for the deviation of the empirical nearest neighbour spacing distribution and
the limiting random matrix prediction was displayed, and shown to have structure.
 To explain the functional form of this correction, in keeping
with the work of Keating and Snaith, Bogomolny and collaborators \cite{BBLM06} were able to present analytic evidence
that the correction term could again be understood in terms of a model of eigenvalues from Haar distributed U$(N)$ matrices.

More explicitly, let $p^{{\rm U}(N)}(0;s)$ refer to the spacing distribution between consecutive eigenvalues for Haar distributed unitary
random matrices, with the angles of the eigenvalues rescaled to have mean
spacing unity. The study \cite{BBLM06} calls for the functional form of $r_2(0;s)$ in the large $N$ expansion
\begin{equation}\label{1}
  p^{{\rm U}(N)}(0; s ) = p_2(0;s) + {1 \over N^2} r_2(0;s) + \cdots
\end{equation}
Here the subscript ``2" on the RHS is the label from Dyson's three fold way in the case of no time reversal symmetry.
As a straightforward application of the pioneering work of Mehta \cite{Me60} and Gaudin \cite{Ga61},
Dyson \cite{Dy62e} showed that
\begin{equation}\label{2}
p_2(0;s) = {d^2 \over d s^2} \det (\II - \KK_{s}),
\end{equation}
where $ \KK_{s}$ is the integral operator on $(0,s)$ with kernel
\begin{equation}\label{3a}
K(x,y) = {\sin \pi (x - y) \over \pi (x - y)}.
\end{equation}
Twenty years later this Fredholm determinant evaluation was put in the context of Painlev\'e theory by the
Kyoto school of Jimbo et al.~\cite{JMMS80}, who showed
\begin{equation}\label{KS}
\det ( \II - \xi \KK_s) = \exp \int_0^{\pi s} {\sigma^{(0)} (t;\xi) \over t} \, dt,
\end{equation}
where $\sigma$ satisfies the differential equation (an example of the Painlevé V equation in sigma form)
\begin{equation}\label{d1}
(t \sigma'')^2 + 4 (t \sigma' - \sigma)(t \sigma' - \sigma + (\sigma')^2) = 0
\end{equation}
with small $t$ boundary conditions 
\begin{equation}\label{d2}
\sigma^{(0)}(t;\xi) = - {\xi \over \pi} t - {\xi^2 \over \pi^2} t^2 + \O(t^3).
\end{equation}
Note that the parameter $\xi$ introduced in (\ref{KS}) only enters in the characterisation through the boundary condition;
we refer to~\cite[Ch.~8]{Fo10} for background theory relating to the Painlev\'e equations as they occur in
random matrix theory. 

In \cite{BBLM06} the leading correction term $r_2(0;s)$ was estimated using a numerical extrapolation from the
tabulation of $p^{{\rm U}(N)}(0; s ) - p_2(0;s)$. In \cite{FM15} the task of finding
analytic forms for $r_2(0;s)$ was addressed. One first notes that
analogous to the formula (\ref{2}), for finite $N$, 
\begin{equation}\label{1a}
 p^{{\rm U}(N)}(0; s ) =  {d^2 \over d s^2} \det (\II - \KK_{s}^N),
\end{equation}
where $ \KK_{s}^N$ is the integral operator on $(0,s)$ with kernel
\begin{equation}\label{3}
K^N(x,y) =   {\sin \pi (x - y) \over N \sin (\pi (x-y)/N)}.
\end{equation}
Expanding now for large $N$ gives
\begin{equation}\label{eq:KNexpansion}
K^N(x,y) = K(x,y) + \frac1{N^{2}} L(x,y) + \O \Big ( {1 \over N^4} \Big ),
\end{equation}
where the leading correction term in this expansion is given explicitly by
\begin{equation}\label{L}
L(x,y) = (\pi(x-y)/6) \sin \pi (x - y).
\end{equation}
It was observed in \cite{FM15} (cf. also  Lemma~\ref{lem:folklore} below) that this expansion lifts to
\begin{equation}\label{Tq0}
\det ( \II -  \KK_s^N) = \det ( \II -  \KK_s) +  \frac{1}{N^2} \Om{\KK_s}{\LL_s} + \O \Big ( {1 \over N^4} \Big ),
\end{equation}
where $\KK_s$ is as in (\ref{3a}) and $\LL_s$ is the integral operator on $(0,s)$ with kernel $L$;
the leading correction is now given by the operator expression
\begin{equation}
\Om{\KK_s}{\LL_s} = -\det(\II - \KK_s) \tr \left((\II - \KK_s)^{-1}\LL_s \right),
\end{equation}
which is {\em linear} in $\LL$.
Substituting in (\ref{1a}) and comparing to (\ref{1}) we read off that
\begin{equation}\label{Tq1}
r_2(0;s) =    {d^2 \over d s^2} \Om{\KK_s}{\LL_s}.
\end{equation}
In Section~\ref{sub:numerical} of the present paper the problem addressed is how to use such formulae to provide a high precision numerical tabulation of
$r_2(k;s)$. In the case $k=1$ we exhibit the resulting functional form in Odlyzko's data set for the Riemann zeros; see Section \ref{sect:riemann}.

An expression for the large $N$ expansion of the Fredholm determinant in (\ref{1a}) involving Painlev\'e transcendents was also given
in \cite{FM15}. Thus it was shown that
\begin{multline}\label{2.13d}
\det ( \II - \xi \KK_s^N) \\ = 
\exp \left( \int_0^{\pi s} {\sigma^{(0)}(t;\xi) \over t} \, dt \right)
\left( 1 + {1 \over N^2} \int_0^{\pi s} {\sigma^{(1)}(t;\xi) \over t} \, dt + \O \Big ( {1 \over N^4} \Big ) \right ).
\end{multline}
As noted above, $\sigma^{(0)}(t;\xi)$  satisfies the particular Painlev\'e V equation in sigma form
(\ref{d1}) with boundary condition (\ref{d2}), with only the latter depending on $\xi$. The  function $\sigma^{(1)}(t;\xi)$ was characterised as the solution of
 the second order, linear differential equation
\begin{equation}\label{Sig2}
A(t;\xi) y'' + B(t;\xi) y' + C(t;\xi) y  = D(t;\xi),
\end{equation}
where, with $\sigma = \sigma^{(0)}(t;\xi)$,
\begin{align}
\nonumber A(t;\xi) & = 2 t^2 \sigma'',\\
\nonumber B(t;\xi) & = - 8  \sigma' \sigma + 12 t (\sigma')^2 + 8t \left( t \sigma' - \sigma \right),\\
\nonumber C(t;\xi) & = - 4 (\sigma')^2 - 8 \left( t \sigma' - \sigma \right), \\
\nonumber D(t;\xi) & = -\frac{4} {3} t^2 \sigma'' \left( \sigma- t  \sigma' - { t^2 \over 2} \sigma'' \right) \\
\nonumber & \quad - \frac{4} {3} (t  \sigma' - \sigma  )\left( 3 \sigma^2 + 2t \sigma \left( t- \sigma'  \right)- 2t^2 \sigma'  \left( t+ \sigma'  \right) \right),
\end{align}
and fulfilling the $t \to 0^+$ boundary condition 
\begin{align}\label{bcs}
\sigma^{(1)}(t;\xi) & = - \Big ( t^4 {\xi^2 \over 9 \pi^2} + t^5 {5 \xi^3 \over 36 \pi^3} + \O(t^6) \Big ).
\end{align}
Substituting (\ref{2.13d}) in  (\ref{1a}) and comparing to (\ref{1}) we read off that
\begin{equation}\label{Tq2}
r_2(0;s) =    {d^2 \over d s^2} \exp \left( \int_0^{\pi s} {\sigma^{(0)}(t;1) \over t} \, dt \right)  \int_0^{\pi s} {\sigma^{(1)}(t;1) \over t} \, dt.
\end{equation}
By using a nested power series method to solve the differential equation (\ref{Sig2}) numerically, 
this expression was used in  \cite{FM15} to determine the graphical form of $r_2(0;s)$, and also to determine
its small and large $s$ expansions. However, the operator approach advocated in the present paper is simpler
and more straightforward to use numerically.

Random unitary matrices from U$(N)$ chosen with Haar measure form one of Dyson's three circular ensembles
of unitary random matrices coming from the considerations of the three fold way, and in this context is
referred to as the circular unitary ensemble (CUE).
The other two circular ensembles are the circular orthogonal ensemble (COE) of symmetric unitary matrices,
and the circular symplectic ensemble (CSE) of self dual unitary matrices.
We will show that our methods can be adapted to provide a systematic
investigation of the leading correction in the large $N$ expansion of
spacing distributions for each of these examples. 

\subsection{Outline of the paper}
After some preparatory material on the expansion of operator determinants and numerical
methods for evaluating the corresponding terms in Section~\ref{sect:prep}, we begin in Section~\ref{sect:CUE} by extending the study initiated in \cite{FM15} on this problem as it applies to CUE matrices. 
The first question addressed is the computation of the leading order correction term in the large $N$
expansion of $p^{{\rm U}(N)}(k; s ) - p_2(k;s)$, or equivalently
$p^{{\rm CUE}_N}(k; s ) - p_2(k;s)$, in terms of an integral operator formula extending
(\ref{Tq1}). We then do the same in the circumstance that each eigenvalue has been deleted independently
at random with probability $(1 - \xi)$. Such a thinning procedure, well known in the theory of point
processes (see e.g.~\cite{IPSS08}), was introduced into random matrix theory by Bohigas and Pato \cite{BP04},
and has been applied to Odlyzko's data set for the Riemann zeros in \cite{FM15}. 
Continuing with the point process perspective, next we consider the setting of a translationally invariant one-dimensional point process, normalised to have average spacing unity, superimposed with a Poisson point process.
Our interest is in the distribution of the minimum distance
from each Poisson point to a point in the original process.
In an ensemble setting, this is equivalent to computing the minimum distance 
from the origin of
a point in the original process. 
For CUE matrices the distribution can be expressed in a form analogous to \eqref{1a}, which allows the statistic to be expanded
for large $N$.
In the final subsection of Section~\ref{sect:CUE}, the statistic for the nearest neighbor spacing, that is, the minimum of the spacing distance between left and right neighbours
in the CUE is similarly studied. In Section~\ref{sect:riemann} our results are compared against
empirical findings from Odlyzko's data set for the Riemann zeros. There we will use a unique $\O(N^{-3})$ correction of
the CUE correlation kernel that was observed in
the study \cite{BBLM06} to match the corresponding $\O(N^{-3})$ terms of the two- and three-point correlation functions. Including this term in the numerical calculations systematically improves the fit to the empirical data.

The study of spacing distributions for matrices from the COE and CSE is more involved than for CUE matrices. The reason
is that the latter is based on the single integral operator $\mathbb{K}_{s}^N$ with kernel (\ref{3}), whereas the pathway to tractable expressions
in the case of the COE and CSE makes use of inter-relations between gap probabilities. The necessary theory is
covered in the first part of Section~\ref{sect:COECSE}.
In the second part of Section~\ref{sect:COECSE} we give the leading correction term for the large $N$ expansion of 
\[
p^{{\rm C}\beta {\rm E}}_\xi(k; s ) - p_{\beta, \xi}(k;s),
\]
where $\beta = 1$ for the COE and $\beta = 4$ for the CSE, in terms of a characterisation analogous to 
(\ref{2.13d}) and (\ref{Sig2}). The parameter $\xi$ specifies the thinning probability.

\subsection{A note on sampling sizes of empirical data}
In this paper we discuss finite size effects up to an error $\O(N^{-4})$, with $N$ in the Odlyzko data set
of Riemann zeros being effectively $N\approx 10$. We choose the same order of magnitude of $N$ for the simulations of the
circular ensembles. Now, by the law of large numbers, sampling is known to introduce a statistical error of about $\O(M^{-1/2})$
where $M$ denotes the sample size. Hence, pushing the sampling error to the same order of magnitude as the remaining
finite size error thus requires a sampling size of $M = N^8 \approx 10^8$. This was actually the choice for our simulations 
and is well matched by the Odlyzko data set of a little more than $10^9 \approx 10^8\cdot N$ zeros. 
However, observing structure also in the $\O(N^{-4})$ remainder term would hence require to increase the sampling size by at least
two to four orders of magnitude. This is only possible with parallel processing and a subtle memory management of the
resulting gigantic data sets (the raw material provided by Odlyzko is already about 22GB in size).

\section{Expansions of determinants and their numerical evaluation}\label{sect:prep}
\subsection{Expansions of operator determinants}

The integral operator formulae of this paper are based on the following {\em folklore} lemma, which we prove here for 
ease of reference.  

\begin{lemma}\label{lem:folklore} Let $J$ be a bounded interval and let $K_h(x,y)$, $K(x,y)$, $L(x,y)$ be continuously differentiable kernels on $J\times J$. If the expansion
\[
K_h(x,y) = K(x,y) + h L(x,y) + O(h^2)
\]
holds uniformly up to the first derivatives as $h\to 0$, then it lifts as an expansion
\[
\KK_h = \KK + h \LL + O(h^2)
\]
of the induced integral operators on $L^2(J)$ in trace-class norm. If $1\not\in \mathrm{Spec} (\KK)$, the operator determinant expands as
\begin{subequations}
\begin{equation}
\det(\II - \KK_h) = \det(\II - \KK) + h \Om{\KK}{\LL}  + O(h^2)
\end{equation}
with the leading correction term given by\footnote{Note that $\Omega(\KK) = -\det(\II-\KK)(\II-\KK)^{-1}$
is the Fréchet derivative of the map $\KK \mapsto \det(\II-\KK)$ under the identification of the dual of the space of
trace-class operators with the space of bounded operators, see \cite[Cor.~5.2]{Si05}. The derivative 
of that map in the direction of
the bounded operator $\LL$ is thus given by $\Om{\KK}{\LL} = \tr(\Omega(\KK)\LL)$. By analytic continuation this interpretation of the expression $\Om{\KK}{\LL}$ extends well to the case that $1\in \mathrm{Spec}(\KK)$.}
\begin{equation}
\Om{\KK}{\LL} = -\det(\II - \KK) \tr \left((\II - \KK)^{-1}\LL \right),
\end{equation}
\end{subequations}
an expression that is linear in $\LL$.
\end{lemma}
\begin{proof}
On {\em bounded}\/ intervals $J$, continuously differentiable kernels $K(x,y)$ induce integral operators $\KK$ on $L^2(J)$ that are trace-class: using the fundamental theorem of calculus one can represent them straightforwardly as the sum of a rank-one operator and a product of two Hilbert--Schmidt operators, see, e.g., \cite[p.~879]{Bo08}. By the same construction the
expansion lifts from kernels in $C^1$-norm to operators in trace-class norm. Now, using the multiplicativity of the operator determinant
we get first that
\[
\det(\II - \KK_h) = \det(\II - \KK) \det\left(\II - h(\II - \KK)^{-1}\LL + \O(h^2)\right).
\]
Truncating the series definition of the operator determinant in terms of exterior products, see, e.g., \cite[Eq.~(3.5)]{Si05},
\[
\det(\II + z \TT) = 1 + \sum_{n=1}^\infty z^n \tr \textstyle\bigwedge^n \TT,
\]
 at the second term we get next that
\[
\det\left(\II - h(\II - \KK)^{-1}\LL + \O(h^2)\right)  = 1- h \tr \left((\II - \KK)^{-1}\LL )\right) + \O(h^2),
\]
which completes the proof.
\end{proof}

\subsection{Numerical evaluation of operator terms}\label{sub:numerical} 
The computations for all figures of this paper are based on the numerical
evaluation of operator terms such as $\det(\II - \KK)$ or $\Om{\KK}{\LL}$. In fact, there is an
extension of the numerical 
method introduced in \cite{Bo08} for the operator determinant that facilitates the efficient highly-accurate evaluation of such terms
with exponential convergence\footnote{There is a constant $c>0$ such that
the approximation error is $\O(e^{-cn})$ where $n$ denotes the dimension parameter of the
method (e.g., the number of quadrature points).} if the kernels extend analytically into the complex domain.
Given a quadrature
method
\[
\int_J f(x) \,dx \approx \sum_{j=1}^n f(x_j) w_j 
\]
with {\em positive} weights $w_j$, we associate to an integral operator $\KK$ on $L^2(J)$ with the kernel $K(x,y)$ the {\em Nyström matrix}
\[
\KK_w = \Big(K(x_j,x_k) w_k\Big)_{j,k=1}^n.
\]
For compact intervals $J$ the best choice for smooth $f$ is Gauss--Legendre quadrature, for which
recently a super-fast algorithm of $\O(n)$ complexity has been found \cite{Bg14}; it converges
exponentially fast if $f$ extends analytically into the complex domain.

As was shown in  \cite[Thm.~6.2]{Bo08}, these convergence properties are inherited by the determinant of the
Nyström matrix as an approximation of the operator determinant, namely
\[
\det(\II-\KK) \approx \det(\II - \KK_w),
\]
where $\II$ denotes the identity operator in $L^2(J)$ on the left and the $n\times n$ identity matrix on the right. In particular, the convergence is exponential if the kernel $K$ extends analytically into the complex domain. Now, the same method of replacing the integral operator by the associated Nyström matrix applies to the approximation of
quite general operator terms. We give four examples:
\begin{itemize}
\item The trace is approximated straightforwardly by the quadrature formula as
\[
\tr(\KK) = \int_J K(x,x)\,dx \approx \sum_{j=1}^n K(x_j,x_j) w_j = \tr(\KK_w).
\]
\item Operator products are approximated by matrix products, that is,
\[
(\KK \LL)_w \approx \KK_w \LL_w.
\]
This follows from
\begin{align*}
((\KK\LL)_w)_{jk} &= w_k \int_J K(x_j,x') L(x',x_k) \,dx' \\
&\approx \sum_{l=1}^n K(x_j,x_l) w_l L(x_l,x_k) w_k
= \sum_{l=1}^n (\KK_w)_{jl} (\LL_w)_{lk}.
\end{align*}

\item The resolvent kernel $R(x,x')$ of the operator $(\II - \KK)^{-1}$ satisfies the integral equation
\[
u(x) = \int_J R(x,x') u(x') \,dx'  - \int_J\int_J K(x,z) R(z,x') u(x')\,dx' dz
\]
for all $u$. Application of the quadrature rule gives with $u_j=u(x_j)$
\[
u_j \approx \sum_{k=1}^n R(x_j,x_k) w_k \cdot u_k - \sum_{l,k=1}^n K(x_j,x_l) w_l\cdot R(x_l,x_k) w_k \cdot u_k,
\]
that is $\II \approx (\II-\KK_w) \left((\II-\KK)^{-1}\right)_w$ and therefore $\left((\II-\KK)^{-1}\right)_w\approx (\II- \KK_w)^{-1}$.
\item By combining all the examples so far, we get 
\begin{align*}
\Om{\KK}{\LL} &= -\det(\II - \KK) \tr \left((\II - \KK)^{-1}\LL \right) \\
 &\approx -\det(\II - \KK_w) \tr \left((\II - \KK_w)^{-1}\LL_w\right) = \Om{\KK_w}{\LL_w}. 
\end{align*}
\end{itemize}
In all the above the convergence properties are inherited from the underlying quadrature formula; in particular, 
the convergence is exponential if the kernels extend analytically to the complex domain. The proof is straightforward in the first two examples and follows from the theory of the
Nyström method for integral equations in the third one; the fourth is a combination of all the previous results.

The numerical derivatives with respect to the $s$ and $z$ variables of the generating functions for the spacing distributions are computed in exactly the same fashion as discussed in \cite{Bo09}, that  is, based
on Chebyshev expansions with respect to $s$ and contour integration with respect to $z$.

The actual implementation and use of this methodology within the Matlab toolbox of \cite{Bo09} is described in the appendix
of the present paper.

\section{Spacing distributions for the CUE}\label{sect:CUE}
\subsection{Preliminaries}
The sequence of angles, to be referred to as eigen-angles, specifying the eigenvalues of a unitary random matrix
from any of the three circular ensembles forms a point process on $(-\pi, \pi]$ with uniform density $N/2\pi$.
Due to the rotational symmetry, 
the spacing distributions of the angles, $p^{{\rm C}\beta{\rm E}}(k;s)$
with $\beta = 1,2, 4$ for the COE, CUE, CSE respectively, and $k=0,1,\ldots$, are thus independent of the absolute location of the eigenvalues.
If there is an eigen-angle at $x$ and at $x+s$, we may rotate all the
angles so that $x$ becomes the origin, and speak of a spacing of size $s$.
Let us do this, and also normalise the angles so that the mean spacing is unity.

Fundamental to the study of spacing distributions 
$p^{{\rm C}\beta{\rm E}}(k;s)$
are the so-called gap probabilities $E^{{\rm C}\beta{\rm E}}(l;s)$.
The latter specify the probability
that an interval of size $s$ contains exactly $l$ eigen-angles for ${\rm C}\beta{\rm E}$ matrices.
Specifically, let us introduce the generating functions
\begin{align}
p^{{\rm C}\beta {\rm E}}(s;z) & := \sum_{n=0}^N (1 - z)^n p^{{\rm C}\beta {\rm E}}(n;s), \label{G1} \\
E^{{\rm C}\beta {\rm E}}(s;z) & := \sum_{n=0}^N (1 - z)^n E^{{\rm C}\beta {\rm E}}(n;s). \label{G2}
\end{align}
Of course the naming on the LHS's represent an abuse of notation, due to the same
functional forms also appearing on the RHS; however the appearance of the generating function symbol $z$ is enough to distinguish the different quantities.
We have the relationship between generating functions (see e.g.~\cite[Prop.~8.1]{Fo10})
\begin{equation}\label{G3}
p^{{\rm C}\beta {\rm E}}(s;z) = {1 \over z^2} {d^2 \over d s^2} E^{{\rm C}\beta {\rm E}}(s;z),
\end{equation}
or equivalently the formula
\begin{equation}\label{G4}
p^{{\rm C}\beta {\rm E}}(n;s) =  {d^2 \over d s^2} \sum_{j=0}^n ( n - j + 1)  E^{{\rm C}\beta {\rm E}}(j;s).
\end{equation}

Also fundamental is the relationship between the generating function (\ref{G2}) and the $k$-point  correlation functions
$\rho_{(n)}^{{\rm C}\beta {\rm E}}(x_1,\dots,x_n)$. Thus we have (see e.g.~\cite[Prop.~9.1.1]{Fo10})
\begin{equation}\label{ER}
E^{{\rm C}\beta {\rm E}}(s;z) = 1 + \sum_{n=1}^N {(-z)^n \over n!} \int_0^s dx_1 \cdots  \int_0^s dx_n \,
\rho_{(n)}^{{\rm C}\beta {\rm E}}(x_1,\dots,x_n).
\end{equation}
The eigenvalues for the CUE form a determinantal point process, while the eigenvalues for the COE and CSE form 
Pfaffian point processes; see e.g.~\cite[Chapters~5 and 6]{Fo10}. 
Specifically, in the case of the CUE we have
\begin{equation}\label{RK}
\rho_{(n)}^{{\rm CUE}}( x_1,\dots, x_n) = \det [ K^N(x_j,x_l)]_{j,l=1,\dots,k},
\end{equation}
where $ K^N$ is given by (\ref{3}), while the analogous formula for the COE and CSE involves a Pfaffian of a $2 \times 2$
anti-symmetric kernel function.
This structural difference distinguishes the case
of the CUE and thus makes it simpler. 

\subsection{Integral operator formulae}\label{S2.2}
Substituting (\ref{RK}) in (\ref{ER}), and making use of a standard expansion formula in the theory
of Fredholm integral operators (see e.g.~\cite[Eq.~(3.14)]{Si05})  gives
\begin{equation}\label{RK1}
E^{{\rm CU E}}( s ;z) = \det ( \II - z \KK_s^N).
\end{equation}
Now substituting the expansion (\ref{Tq0}) in (\ref{RK1}) and substituting the result in (\ref{G3}) we read off the large $N$ expansion
to second order.

\begin{prop}\label{P1}
We have
\begin{equation}\label{RK2}
p^{{\rm CU E}}(s ;z)  = p_2(s;z) + {1 \over N^2} r_2(s;z) + 
\O\Big ( {1 \over N^4}\Big ),
\end{equation}
where
\begin{equation}\label{RK3}
p_2(s;z)  = {1 \over z^2} {d^2 \over d s^2}  \det ( \II - z \KK_s)
\end{equation}
and
\begin{equation}\label{RK4}
r_2(s;z) =  {1 \over z} {d^2 \over d s^2} \Om{z \KK_s}{\LL_s}.
\end{equation}
In particular
\begin{equation}
p^{{\rm CU E}}(0; s )  = p_{2}(s;z)  \Big |_{z=1} + {1 \over N^2} r_{2} (s;z)  \Big |_{z=1} + 
\O\Big ( {1 \over N^4} \Big ),
\end{equation}
and
\begin{equation}
  p^{{\rm CU E}} (1; s )  =- {d \over d z}  p_{2}(s;z) \Big |_{z=1} - {1 \over N^2} {d \over dz} r_{2} (s;z) \Big |_{z=1}+ \O\Big ( {1 \over N^4} \Big ).
\end{equation}
\end{prop}

\subsection{Thinning}\label{S2.3}
The operation of thinning applied to a point process refers to the procedure of independently deleting each point in a sample
with probability $(1 - \xi)$, $0 < \xi \le 1$. With the $k$-point correlation function before thinning being given by $\rho_{(k)}(x_1,\dots,x_k)$,
the $k$-point correlation function after thinning is simply $\xi^k \rho_{(k)}(x_1,\dots,x_k)$,
due to the independence.
 It thus follows from (\ref{ER}), that the generating function for the
gap probability in the presence of thinning, $E_\xi^{{\rm C}\beta{\rm E}}(s;z)$ say, is related to the generating function for the
gap probability without thinning by
\begin{equation}\label{Th}
E_\xi^{{\rm C}\beta{\rm E}}(s;z) = E^{{\rm C}\beta {\rm E}}(s;\xi z).
\end{equation}
Specialising further to the CUE we can substitute (\ref{RK1}) in the RHS of this expression to conclude
\begin{equation}\label{Th1}
E_\xi^{{\rm CUE}}(s;z) =  \det(\II - z \xi \KK_s^N).
\end{equation}
The analogues of (\ref{G3}) and (\ref{G4}) for the spacing probabilities are
\begin{equation}\label{Th2}
p^{{\rm C}\beta {\rm E}}_\xi(s;z) \!= \!{1 \over \xi z^2} {d^2 \over d s^2} E^{{\rm C}\beta {\rm E}}(s;\xi z), \: \: \:
p^{{\rm C}\beta {\rm E}}_\xi(n;s) \!= \!{1\over \xi} {d^2 \over d s^2} \sum_{j=0}^n ( n - j + 1)  E_\xi^{{\rm C}\beta {\rm E}}(j;s).
\end{equation}

Using the first of these, together with  (\ref{RK1}) and the expansion (\ref{Tq0}) we can write down the analogue
of Proposition \ref{P1}.

\begin{prop}\label{P1a}
 We have
\begin{equation}\label{RK5}
p^{{\rm CU E}}_\xi(s;z )  = p_{2, \xi}(s;z) + {1 \over N^2} r_{2, \xi} (s;z) + \O\Big ( {1 \over N^4} \Big ),
\end{equation}
where
\begin{equation}\label{RK6}
p_{2, \xi}(s;z)  = {1 \over \xi z^2} {d^2 \over d s^2}  \det ( \II - \xi z \KK_s)
\end{equation}
and
\begin{equation}\label{RK7}
r_{2,\xi}(s;z) =  {1 \over z} {d^2 \over d s^2} \Om{\xi z \KK_s}{\LL_s}.
\end{equation}
In particular
\begin{equation}
p^{{\rm CU E}}_\xi (0; s )  = p_{2, \xi}(s;1) + {1 \over N^2} r_{2, \xi} (s;1) + \O\Big ( {1 \over N^4} \Big ),
\end{equation}
and
\begin{equation}
p^{{\rm CU E}}_\xi (1; s)  =- {d \over d z}  p_{2, \xi}(s;z) \Big |_{z=1} - {1 \over N^2} {d \over dz} r_{2, \xi} (s;z) \Big |_{z=1}+ \O \Big ( {1 \over N^4} \Big ).
\end{equation}
\end{prop}

Implementation of these formulae according to the method of Section~\ref{sub:numerical} will be carried
out in Section~\ref{sect:riemann}, in the context of a comparison with the corresponding statistics for
Odlyzko's data set of the Riemann zeros.

\begin{figure}
\hspace*{-0.5cm}
\includegraphics[width=0.55\textwidth]{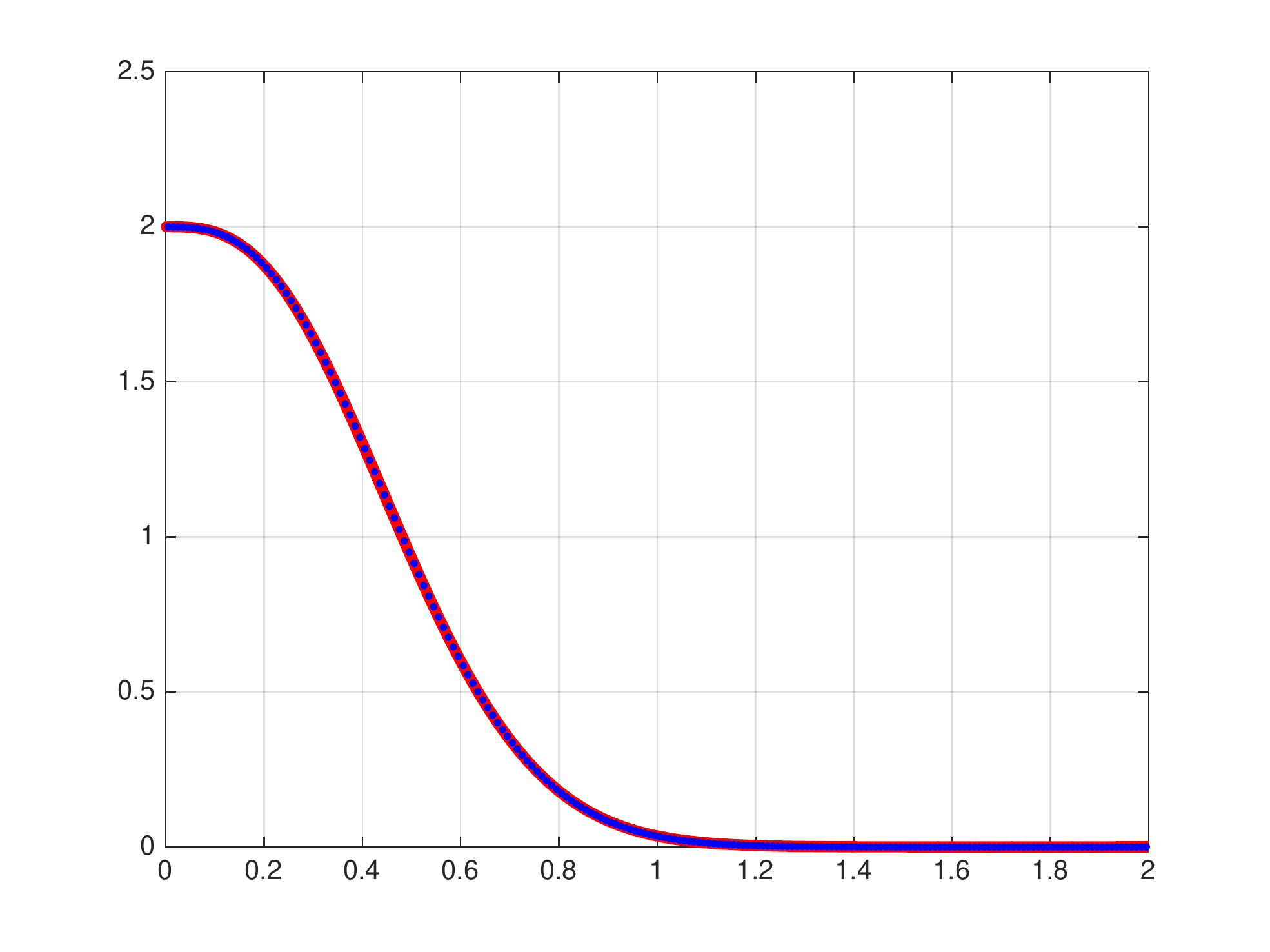}\hspace*{-0.5cm}
\includegraphics[width=0.55\textwidth]{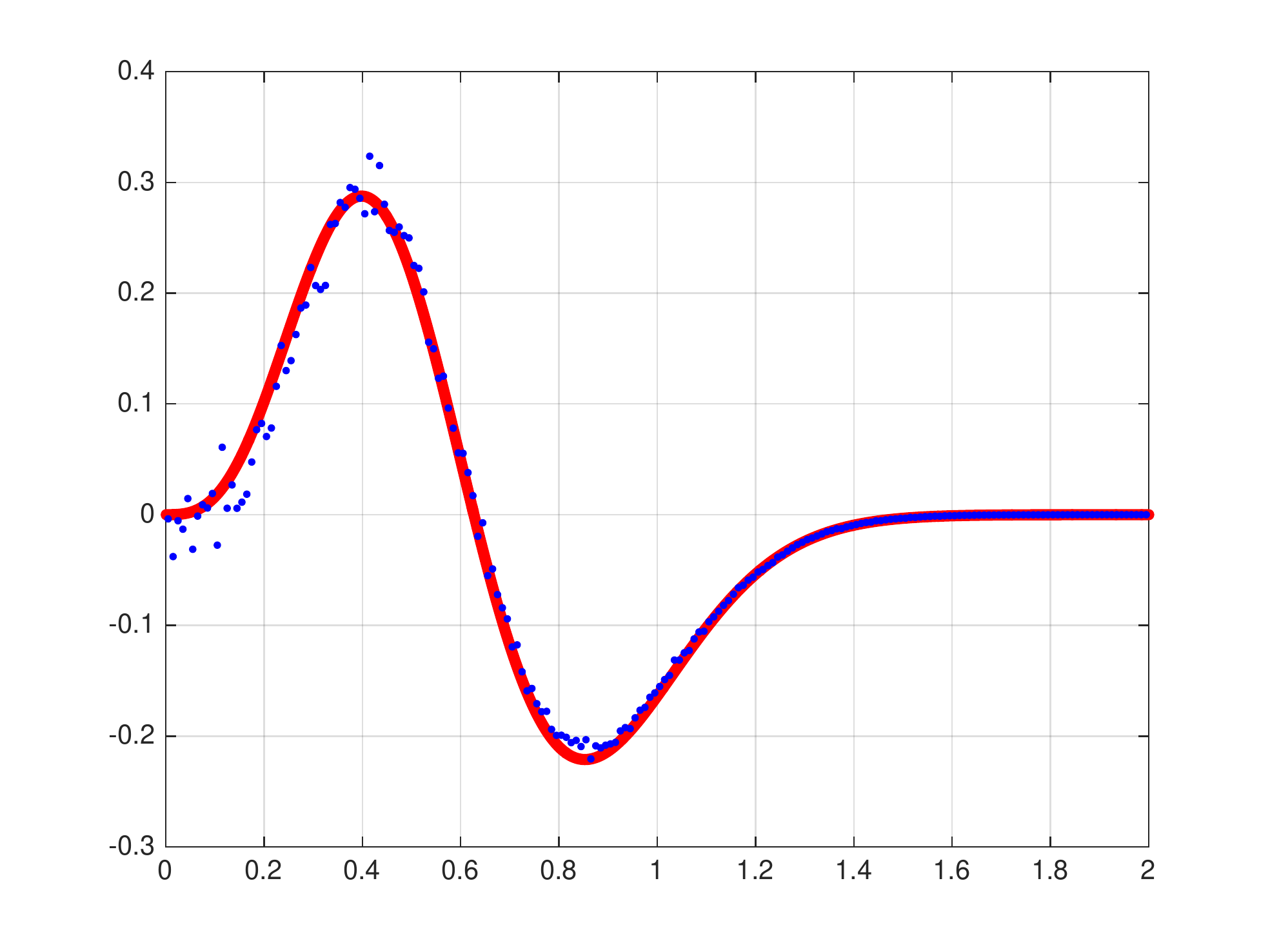}
\caption{Minimum spacing
from a randomly chosen origin: simulation vs. theory for finite size CUE, no thinning ($\xi=1$). Left panel: a histogram of empirical
data from ${\rm CUE}_{N}$ with $N=20$ scaled to unit mean spacing, computed using a bin size of 0.01 and $10^8$ samples from ${\rm CUE}_{N}$, for each of which $1000$ samples of a uniformly distributed random origin were drawn (blue dots); the large $N$ limit $p_{{\rm Origin},\xi}(s)$ (red solid line). Right panel: the simulation data minus $p_{{\rm Origin},\xi}(s)$  scaled
by $N^2$ (blue dots); the leading correction term $r_{{\rm Origin},\xi}(s)$ (red solid line).}\label{fig:5}
\end{figure}

\begin{figure}
\hspace*{-0.5cm}
\includegraphics[width=0.55\textwidth]{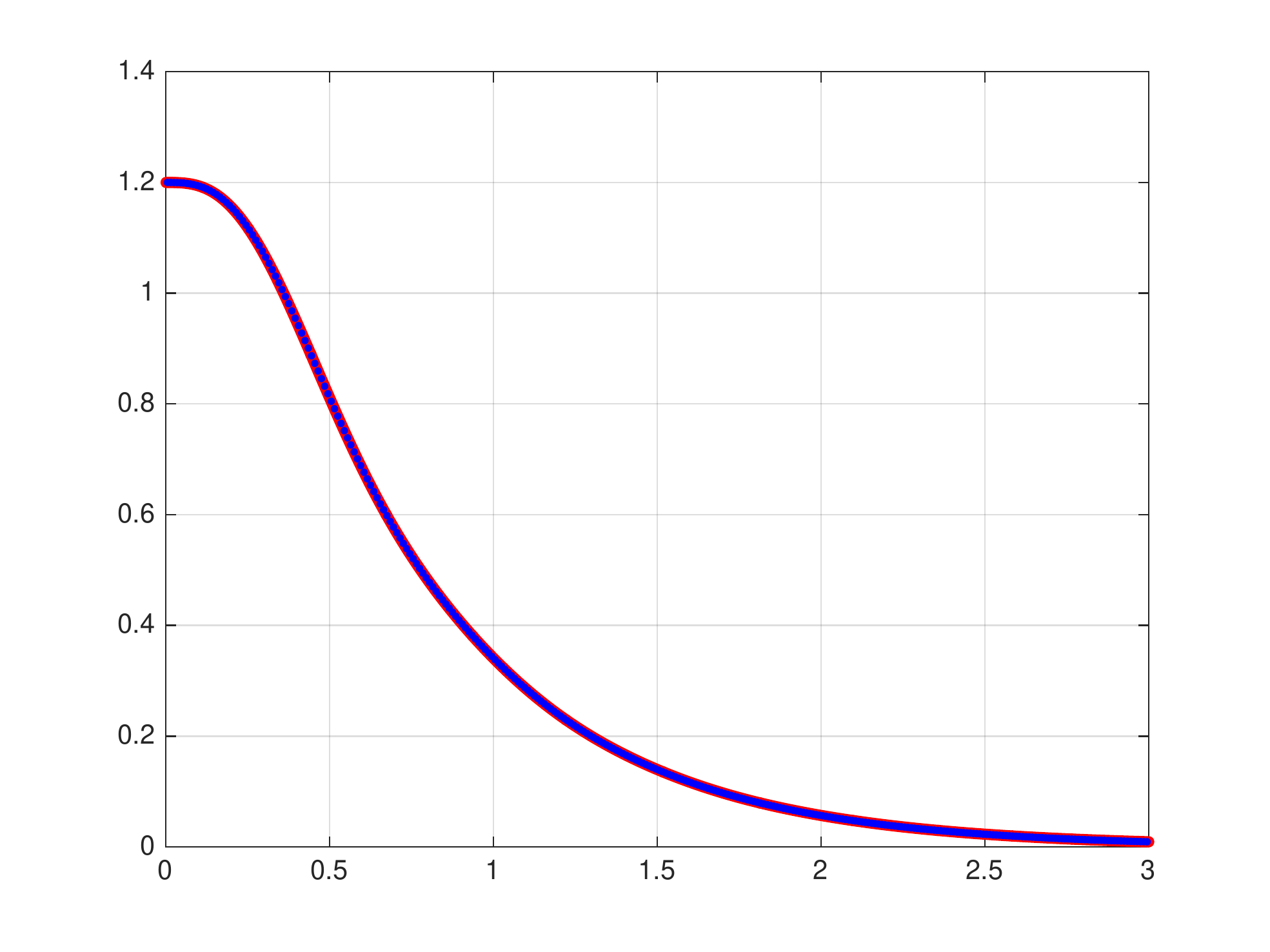}\hspace*{-0.5cm}
\includegraphics[width=0.55\textwidth]{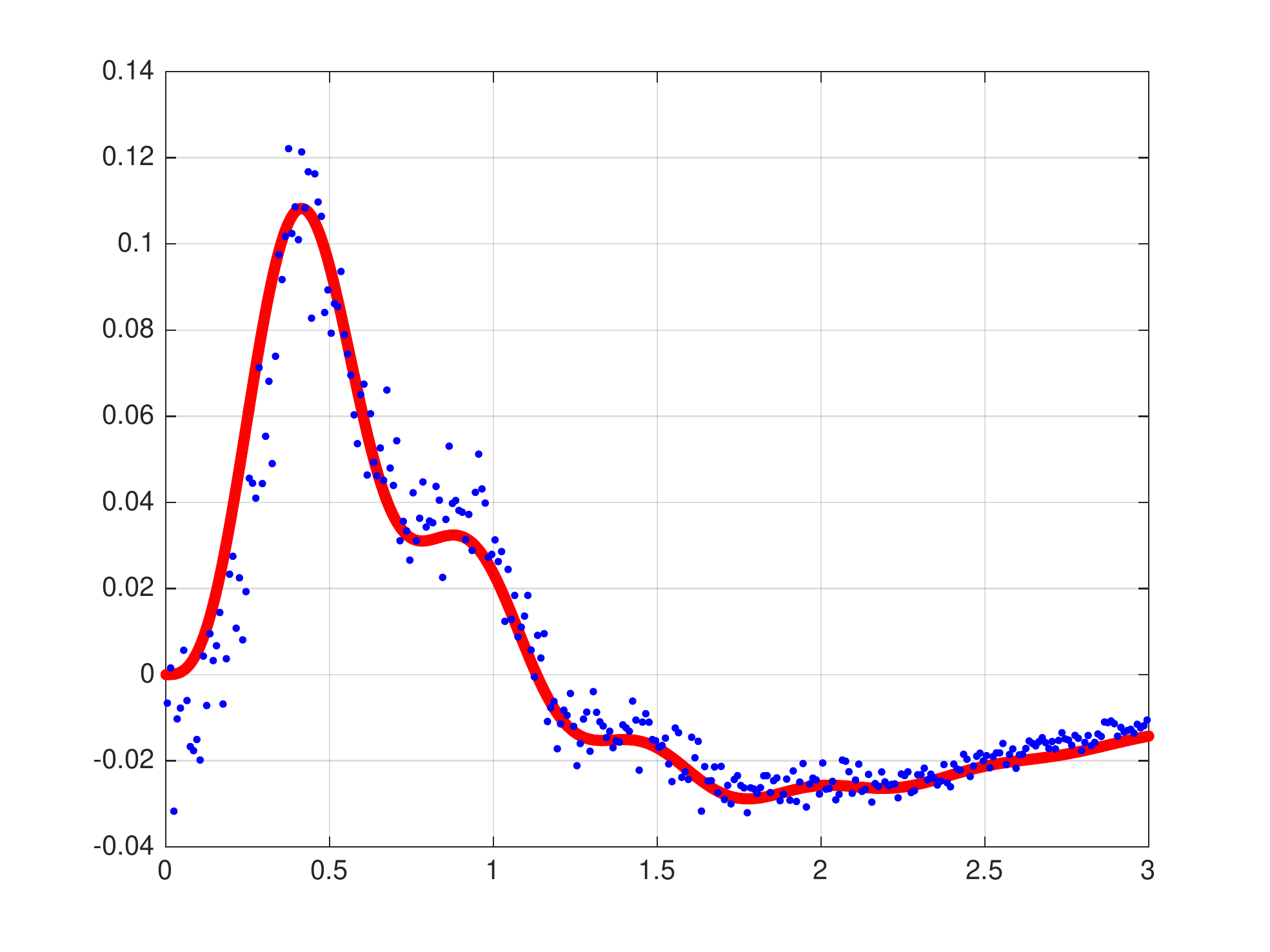}
\caption{As in Figure \ref{fig:5} but with thinning $\xi=0.6$.}\label{fig:6}
\end{figure}

\subsection{Minimum spacing distribution from a randomly chosen origin}\label{sub:origin}
The quantity $E^{{\rm C}\beta{\rm E}}(0;(- s , s ))$ is the probability that there are no eigenvalues in the interval 
$(- s , s )$ of the scaled 
C$\beta$E. Differentiating this quantity with respect to $s$ gives the probability density function
$p_{\rm Origin}^{{\rm C}\beta{\rm E}}(s)$ for the scaled eigen-angle closest to the origin,
$$
p_{\rm Origin}^{{\rm C}\beta{\rm E}}(s) = - {d \over d s} E^{{\rm C}\beta{\rm E}}(0;(-s,s)).
$$
In the case $\beta = 2$ with thinning, it follows from (\ref{Th1}),  as well as the translation invariance of the sine kernel, that
\begin{equation}\label{2.20a}
p_{\rm Origin , \xi}^{{\rm CU E}}(s) = - {d \over d s}  \det ( \II - \xi \KK_{2 s}^N).
\end{equation}
Use of (\ref{2.13d}) gives a characterisation of the leading two terms of the large $N$ expansion of this quantity.

\begin{prop}
We have
\begin{equation}\label{RK10}
p_{\rm Origin , \xi}^{{\rm CUE}}( s )  =  p_{\rm Origin , \xi}(s) + {1 \over N^2} r_{\rm Origin, \xi}(s) + \O\Big ( {1 \over N^4} \Big )
\end{equation}
where in terms of integral operators
\begin{align}\label{2.25x}
 p_{\rm Origin , \xi}(s) & =  - {d \over d s} \det ( \II - \xi \KK_{2 s} ) \nonumber \\
 r_{\rm Origin , \xi}(s) & = -\xi {d \over d s} \Om{\xi\KK_{2s}}{\LL_{2s}}.
\end{align}
Alternatively, in terms of $\sigma^{(0)}$ and $\sigma^{(1)}$ defined in \eqref{2.13d}
\begin{equation}\label{RK11}
 p_{\rm Origin , \xi}(s) =  - {d \over d s}
\exp  \left( \int_0^{2\pi s} {\sigma^{(0)}(t;\xi) \over t} \, dt \right)
\end{equation}
and
\begin{equation}\label{RK12}
 r_{\rm Origin, \xi}(s) = - {d \over d s} 
\Big ( \int_0^{2\pi s} {\sigma^{(1)}(t;\xi) \over t} \, dt  \Big )
\exp  \left( \int_0^{2\pi s} {\sigma^{(0)}(t;\xi) \over t} \, dt \right).
\end{equation}
\end{prop}

To simulate this quantity in translationally invariant empirical data such as the ${\rm CUE}$ itself,
one can actually draw the origin uniformly and independently within the interval for which the data are defined
(that is, one superimposes a Poisson point process defining the origin).
To prepare for corresponding computations within a large set of Riemann zeros this was done, as a proof of concept, for CUE matrices in Figs.~\ref{fig:5} and Figs.~\ref{fig:6}.

\subsection{Nearest neighbour spacing distribution}\label{sub:nn}
A variant of the spacing distribution between consecutive eigenvalues is the spacing distribution between
nearest neighbour eigenvalues \cite{FO96}. For this, at each eigenvalue one measures the smallest of the
spacings to the eigenvalue immediately to the left, and the spacing immediately to the right. In a theoretical ensemble
formulation, the system is to be conditioned so that there is an eigenvalue at the origin. Consider in particular the
CUE with eigen-angles scaled to have unit spacing.
 With $\rho_{(n), \xi}^{{\rm CUE}, \theta = 0}$ denoting the $n$-point correlation function for the
conditioned system in the presence of thinning,
we have
\begin{multline}\label{2.25}
 \rho_{(n), \xi}^{{\rm CUE}, \theta = 0}( \theta_1 , \dots, 
\theta_n )  = \xi^n 
\rho_{(n+1)}^{{\rm CUE}}(\theta_1,\dots, \theta_n ,0)  \\
 = \xi^n \det \left( {\sin \pi (\theta_j - \theta_k) \over N \sin ( \pi (\theta_j - \theta_k)/N)} -
{\sin \pi \theta_j \over N \sin \pi \theta_j/N}{\sin \pi \theta_k \over N \sin \pi \theta_k/N}   \right)_{j,k=1,\dots,n},
\end{multline}
where the second line follows from (\ref{RK}) and (\ref{3}), together with elementary row operations applied to
the determinant.

Let the entry of the determinant in (\ref{2.25}) be denoted by $\tilde{K}^N(\theta_j,\theta_k)$, and denote the corresponding integral operator
supported on $(-s,s)$ by $\tilde{\KK}^N_{(-s,s)}$. Analogous to~(\ref{2.20a}), the scaled nearest neighbour spacing
in the presence of thinning is then
\begin{equation}\label{2.21a}
p_{\rm nn, \xi}^{{\rm CU E}}( s ) = - {d \over d s}  \det \Big ( \II - \xi \tilde{\KK}_{(-s,s)}^N \Big ).
\end{equation}
By observing
\[
\tilde{K}^N(x,y) = K^N(x,y) - K^N(x,0)K^N(0,y)
\]
we can read off from (\ref{2.21a}), upon expanding the kernel to order $1/N^2$, integral operator formulae for the leading
two terms analogous to (\ref{2.25x}).

\begin{prop}
In terms of the notation \eqref{3a} and \eqref{L}, let
\begin{align}\label{KL}
{K}^{\rm nn}(x,y) & = K(x,y) - K(x,0)K(0,y) \nonumber \\  {L}^{\rm nn}(x,y) & = L(x,y) - L(x,0)K(0,y) - K(x,0)L(0,y).
\end{align}
Denote by ${\KK}_{(-s,s)}^{\rm nn}$ and ${\LL}_{(-s,s)}^{\rm nn}$ the integral operators on $(-s, s)$
with kernels ${K}^{\rm nn}(x,y) $ and $ {L}^{\rm nn}(x,y)$ respectively. We have
\begin{equation}\label{2.21b}
p_{\rm nn, \xi}^{{\rm CU E}}(s) =   p_{\rm nn , \xi}(s) + {1 \over N^2} r_{\rm nn, \xi}(s)+ \O\Big ( {1 \over N^4} \Big ) ,
\end{equation}
where
\begin{align}\label{2.25a}
 p_{\rm nn , \xi}(s) & =  - {d \over d s} \det \left( \II - \xi {\KK}_{(-s, s)}^{\rm nn} \right) \nonumber \\
 r_{\rm nn , \xi}(s) & = -\xi {d \over d s} \Om{\xi{\KK}_{(-s, s)}^{\rm nn}}{{\LL}_{(-s, s)}^{\rm nn}}.
\end{align}
\end{prop}

\begin{remark} In \cite{FO96} a Painlev\'e transcendent evaluation of $p_{\rm nn , \xi}(s)$ has been given; see also \cite[\S 9.5.2]{Fo10}.
\end{remark}

\begin{figure}
\hspace*{-0.5cm}
\includegraphics[width=0.55\textwidth]{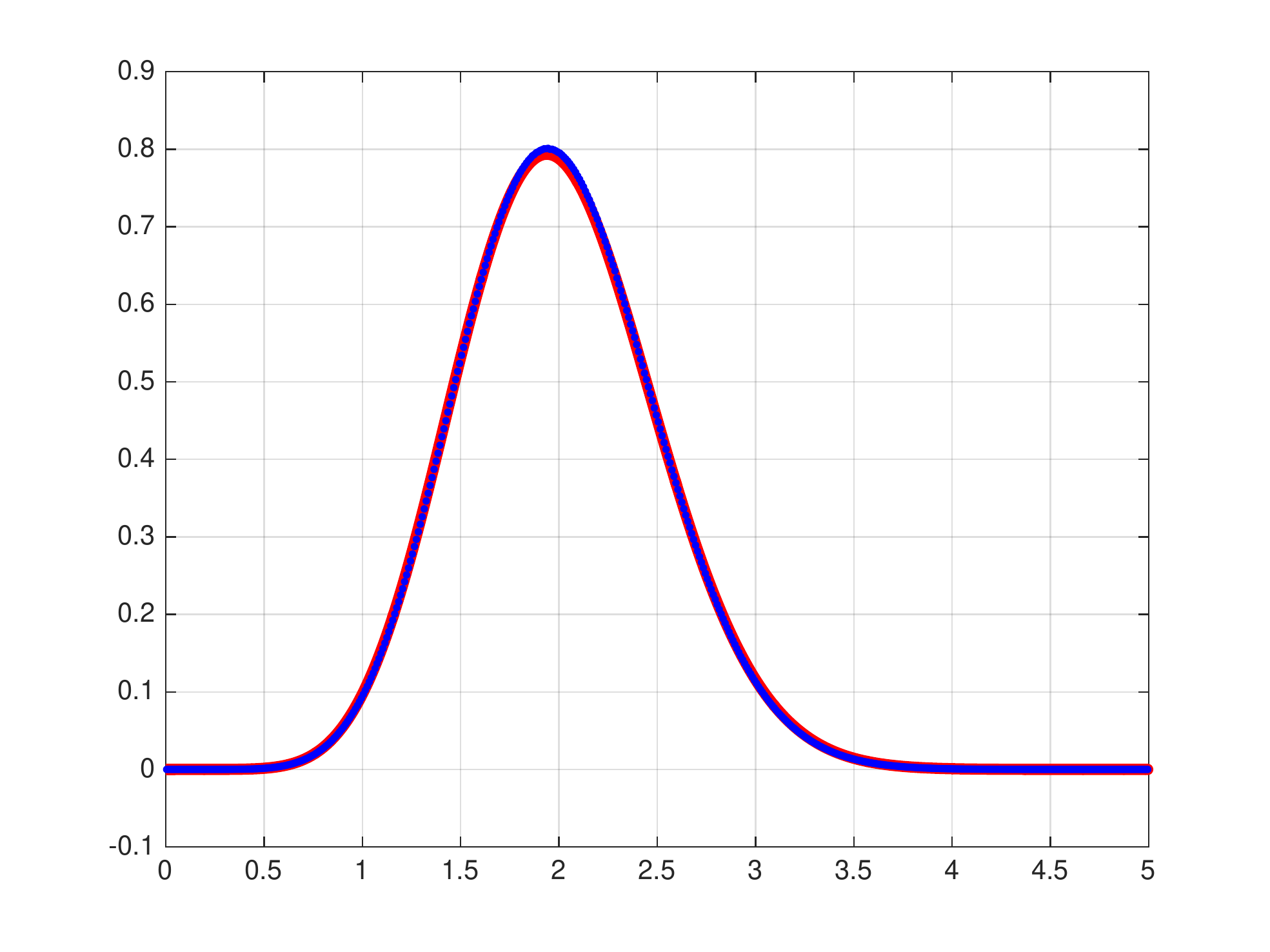}\hspace*{-0.5cm}
\includegraphics[width=0.55\textwidth]{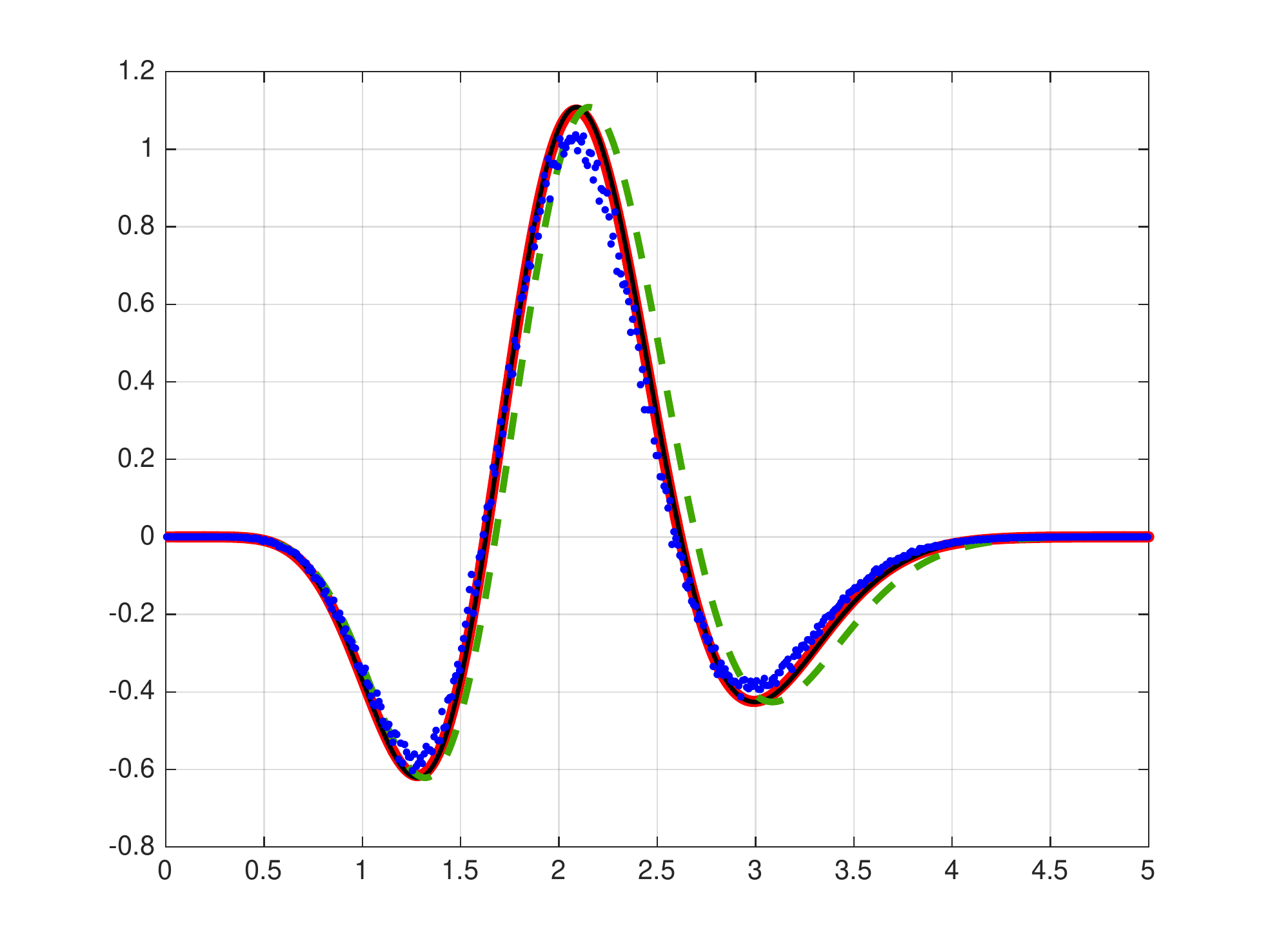}
\caption{$1$-st next neighbour spacing: Riemann zeros data vs. random-matrix based prediction with effective dimension $N=\log(E/2\pi)/\sqrt{12\Lambda}\approx11.3$, no thinning ($\xi=1$). Left panel: a histogram of the Odlyzko data set with bin size $0.01$ (blue dots); 
the large $N$ limit $p_{2,\xi}(1;s)$ (red solid line). Right panel: the Riemann zero data minus $p_{2,\xi}(1;s)$  scaled
by $N^2$ (blue dots); the leading correction term $r_{2,\xi}(1;s)$ with interior rescaling (red solid line), exterior rescaling (thin black line) and no rescaling (dashed green line).}\label{Fm1}\label{fig:1}
\end{figure}

\section{Application to the statistics of Riemann zeros}\label{sect:riemann}

\subsection{The two point correlation function}
Starting with the 1973 seminal work of Montgomery there have been significant developments on a deep (conjectural) connection between
the statistics of the fluctuation properties of the zeros of the Riemann zeta function on the 
critical axis $\Re s = 1/2$ and those of the eigen-angles of ${\rm CUE}_N$, supported by large-scale numerical
calculations based on the extensive tables of Riemann zeros provides by Odlyzko, see, e.g., \cite{Od87,FO96,BK99, KS00a, Od01, BBLM06, FM15} and the literature cited therein. Answering a question of Odlyzko \cite{Od01}, who had posed the challenge 
of understanding the structure in the numerical difference graph between the Riemann zeta spacing distribution for a set 
of zeros of large height and the (conjectured) asymptotics, finite size effects of this statistic were studied in~\cite{BBLM06}. This
study gave a precise quantitative association of Riemann zeros $1/2 + iE$ at height $E$ to a (formal) size $N$ of the corresponding ${\rm CUE}_N$. Using the Hardy--Littlewood conjecture of the distribution of prime pairs, Bogomonly and Keating \cite{BK96b} had earlier given an analytic expression for the pair correlation of the Riemann zeros. The authors of \cite{BBLM06} expanded this for large height $E$, with the local density normalized to unity, to obtain
\begin{equation}\label{eq:RRZ}
R^{\text{RZ}}_2(s) = 1 - \frac{\sin^2(\pi s)}{\pi^2 s^2} - \frac{\Lambda\sin^2 (\pi s)}{\pi^2} \bar\rho^{-2} - \frac{Q s \sin(2\pi s)}{2\pi^2}\bar\rho^{-3}+ \O(\bar\rho^{-4}),
\end{equation}
where
\[
\bar\rho = \frac{\log(E/2\pi)}{2\pi}
\]
denotes the smooth asymptotic density of zeros at $E$ and $\Lambda$, $Q$ are the following constants (with $\gamma_n$ the Stieltjes constants and $\sum_p$ summation over the primes):\footnote{The highly accurate numerical values were obtained by applying the Euler--Maclaurin formula to the Stieltjes constants,
\[
\gamma_n = \lim_{m\to\infty} \left(\sum_{k=1}^m\frac{\log^n k}{k} - \frac{\log^{n+1} m}{m+1}\right),
\]
 and the method of \cite[p.~2]{Ma09} to the sums over the primes.}
\begin{align*}
\Lambda &= \gamma_0^2 + 2 \gamma_1 + \sum_{p}\frac{\log^2 p}{(p-1)^2} = 1.57315\,10713\,24955\,\ldots,\\
Q &= \sum_{p}\frac{\log^3 p}{(p-1)^2} = 2.31584\,63849\,58803\,\ldots.
\end{align*}
On the other hand, the pair correlation function 
\[
R^{\text{CUE}_N}_2(s) = \rho_{(2)}(0,s) = \begin{vmatrix}
K^N(0,0) & K^N(0,s) \\
K^N(s,0) & K^N(s,s) 
\end{vmatrix}
\]
of ${\rm CUE}_N$, normalized to mean spacing unity, expands by \eqref{eq:KNexpansion} as
\begin{equation}\label{eq:RCUE}
R^{\text{CUE}}_2(s) = \rho_{(2)}(x,x+s) = 1 - \frac{\sin^2(\pi s)}{\pi^2 s^2} - \frac{\sin^2 (\pi s)}{3} N^{-2} + \O(N^{-4}).
\end{equation}
By matching the first correction terms of \eqref{eq:RRZ}  and \eqref{eq:RCUE} the authors of \cite{BBLM06} got as effective dimension $N$ of
a ${\rm CUE}_N$ at height $E$
\[
N = \frac{\pi}{\sqrt{3\Lambda}} \bar\rho = \frac{1}{\sqrt{12\Lambda}} \log\left(\frac{E}{2\pi}\right).
\]
This way one has the (conjectural) approximation
\[
R_2^{\text{RZ}}(s) = R_2^{\text{CUE}}(s) + \O(N^{-3}).
\]

\begin{figure}
\hspace*{-0.5cm}
\includegraphics[width=0.55\textwidth]{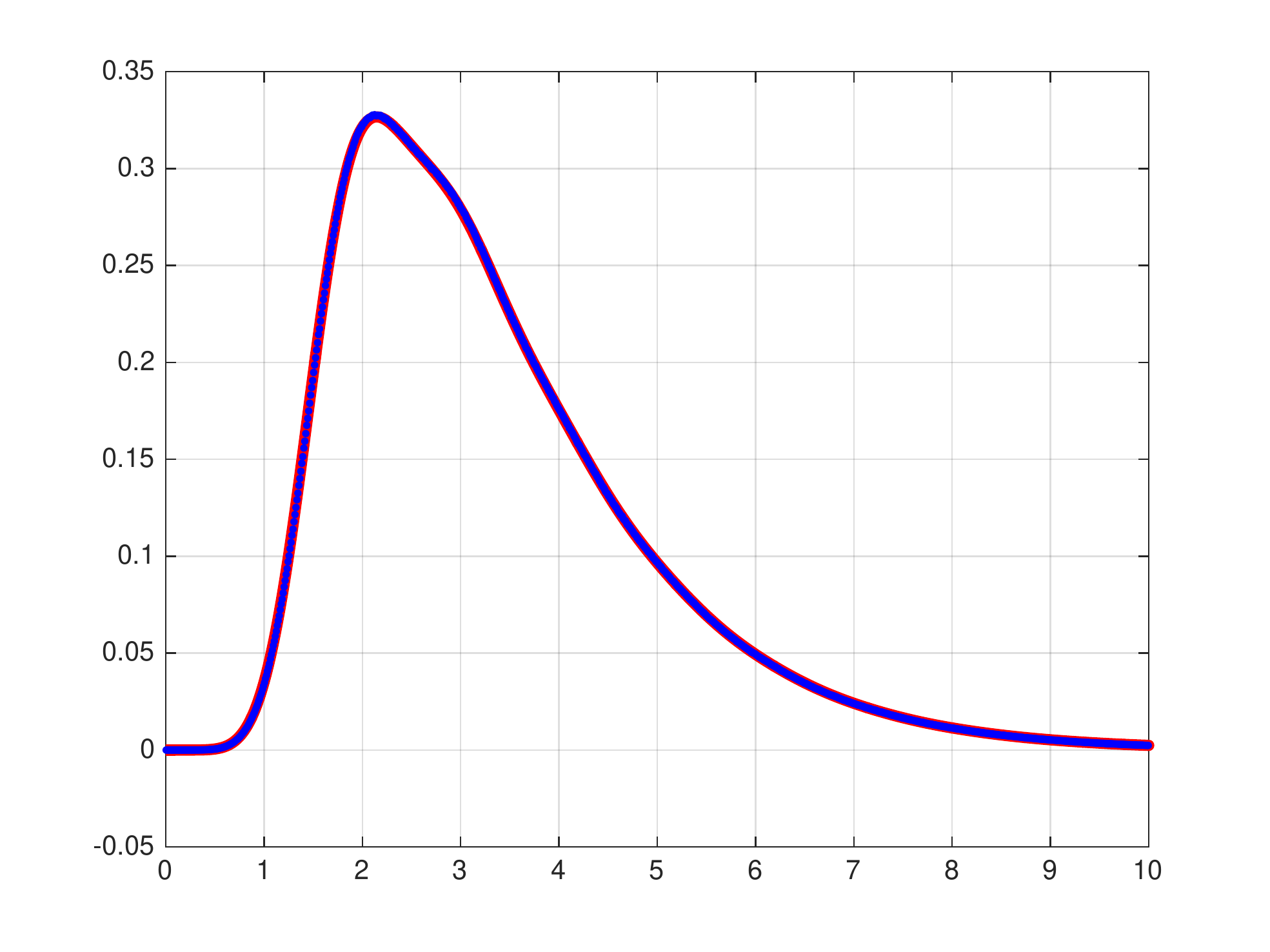}\hspace*{-0.5cm}
\includegraphics[width=0.55\textwidth]{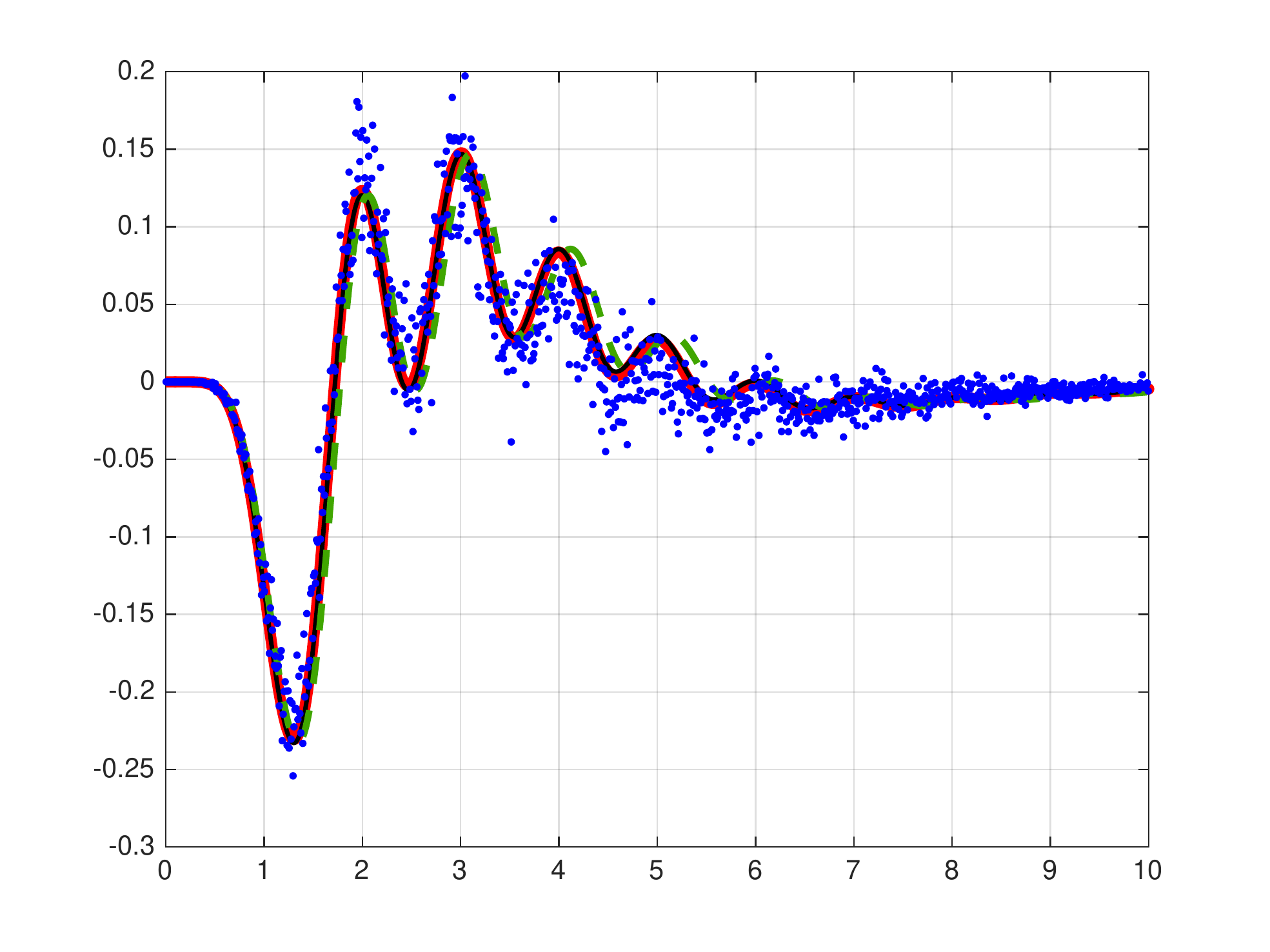}
\caption{As in Figure \ref{fig:1} but with thinning $\xi=0.6$.}\label{fig:2}
\end{figure}

\subsection{Exterior rescaling of the leading correction terms}
As noted in \cite{BBLM06} a rescaling of the $s$-variable in the leading correction term of
the $\bar\rho$-expansion of the two-point functions absorbs the $\O(\bar\rho^{-3})$ term, respectively the $\O(N^{-3})$ term:
\begin{align*}
R^{\text{RZ}}_2(s) &= 1 - \frac{\sin^2(\pi s)}{\pi^2 s^2} - \frac{\Lambda\sin^2 (\pi \alpha s)}{\pi^2} \bar\rho^{-2} + \O(\bar\rho^{-4})\\
&=1 - \frac{\sin^2(\pi s)}{\pi^2 s^2} - \frac{\sin^2 (\pi \alpha s)}{3} N^{-2} + \O(N^{-4}),
\end{align*}
where 
\[
\alpha = 1 + \frac{Q}{2\pi\Lambda\bar\rho} = 1+ \frac{Q}{\Lambda \log(E/2\pi)} = 1 + \frac{Q}{\Lambda\sqrt{12\Lambda}}N^{-1}.
\]
That is, one gets an $O(N^{-4})$ approximation of $R^{\text{RZ}}_2(s)$ by expanding $R^{\text{CUE}}_2(s)$ into powers
of $N^{-2}$ and rescaling the $s$-variable of the leading correction term by $\alpha$. 

It was suggested in \cite{BBLM06} to apply the same procedure, that is, $\alpha$-rescaling the $s$-variable of the leading correction term of a ${\rm CUE}_N$ based expansion, to improve the fit even for more general statistics, such as the 
spacing distribution originally considered by Odlyzko. In \cite{FM15} it was successfully applied to improve
the fit of the CUE leading correction term to the spacing distribution in the presence of thinning.

\begin{figure}[tbp]
\hspace*{-0.5cm}
\includegraphics[width=0.55\textwidth]{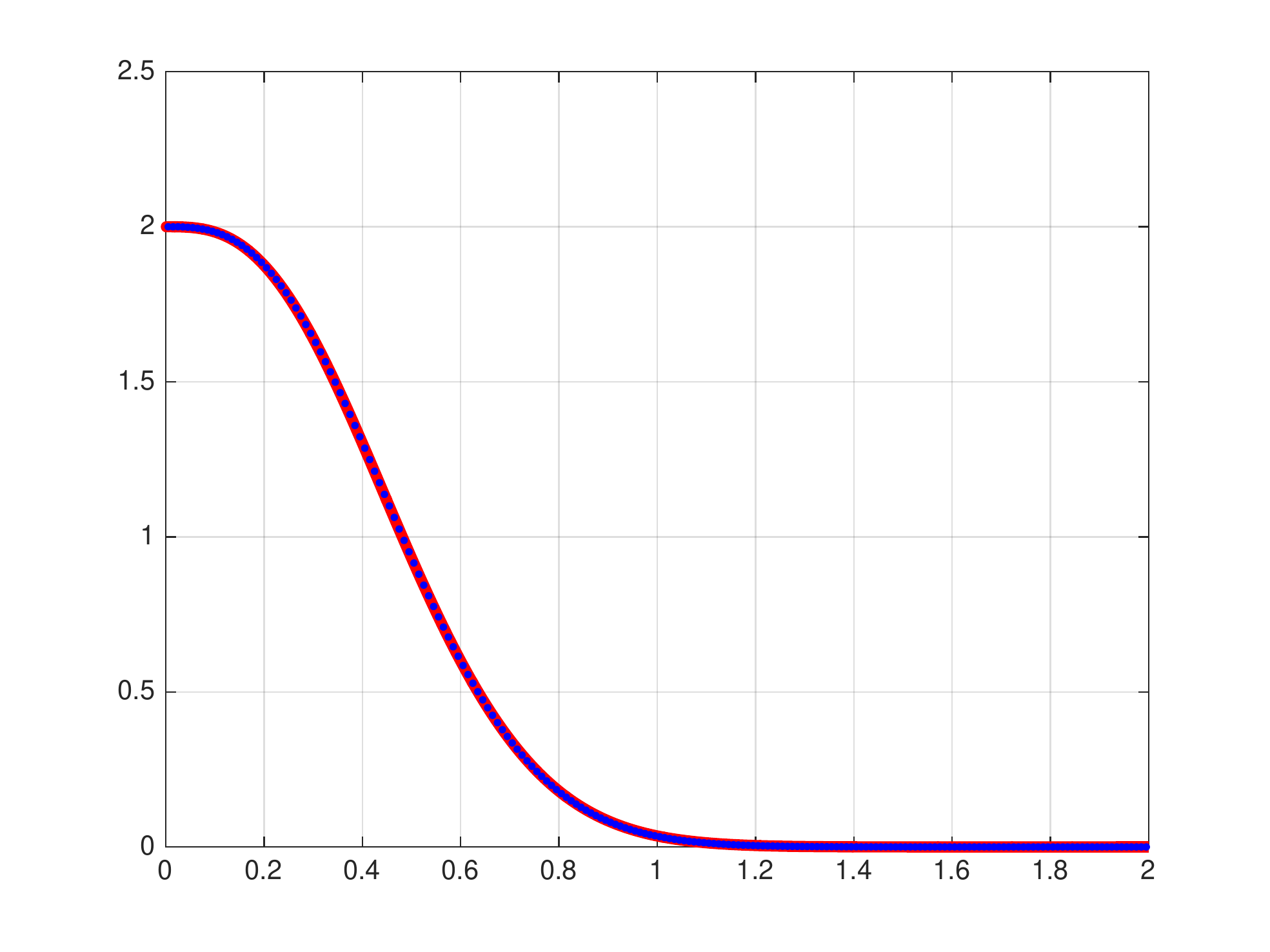}\hspace*{-0.5cm}
\includegraphics[width=0.55\textwidth]{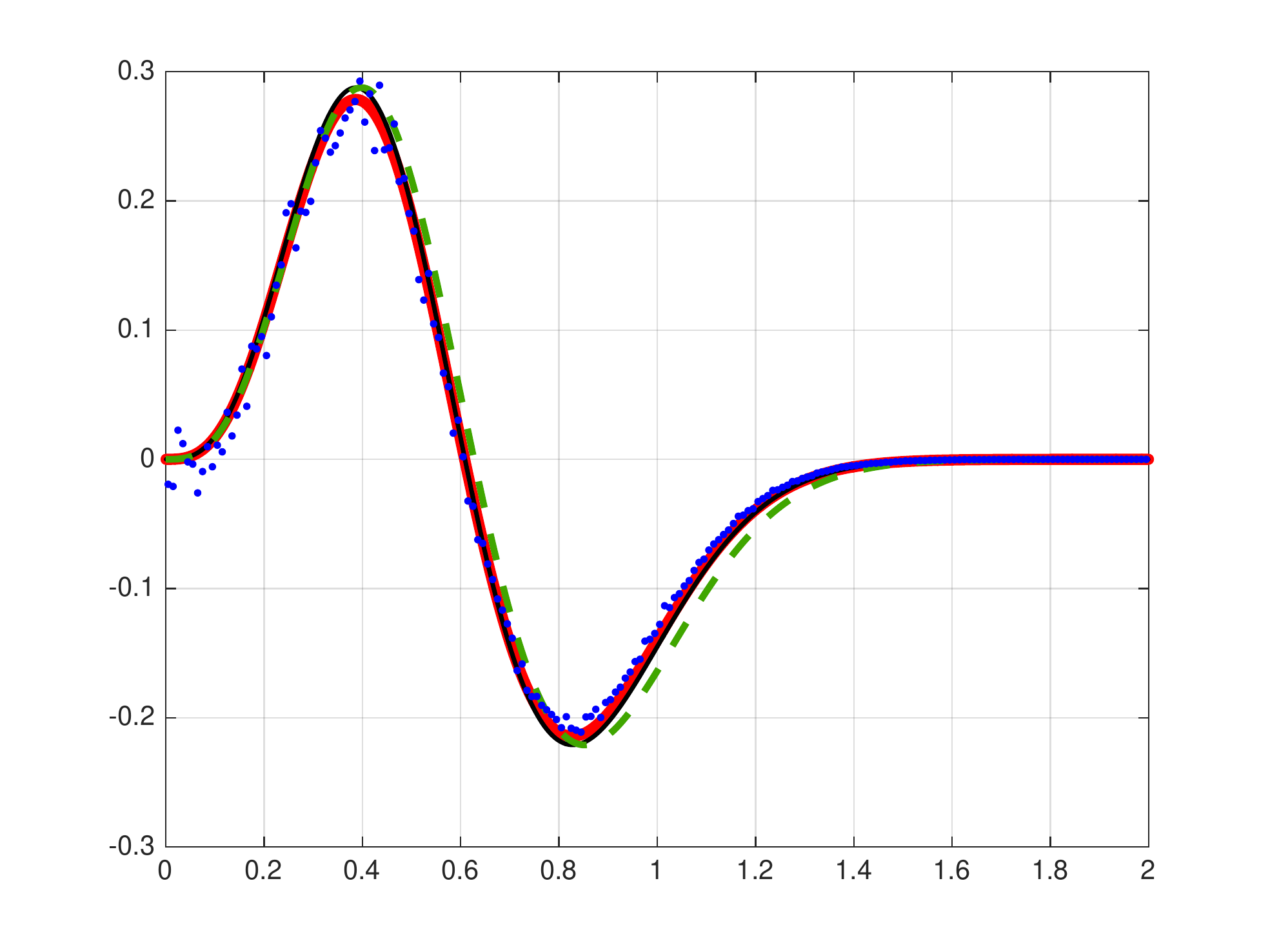}
\caption{Minimum spacing
from a randomly chosen origin: Riemann zeros data vs. random-matrix based prediction with effective dimension $N=\log(E/2\pi)/\sqrt{12\Lambda}\approx11.3$, no thinning ($\xi=1$). Left panel: a histogram of the Odlyzko data set with bin size $0.01$ (blue dots), where
the approximately $10^9$ zeros were broken into 1024 sets of nearly equal size, for each of which $10^7$ samples of a uniformly distributed random origin were drawn; 
the large $N$ limit $p_{{\rm Origin},\xi}(s)$ (red solid line). Right panel: the Riemann zero  data minus $p_{{\rm Origin},\xi}(s)$  scaled
by $N^2$ (blue dots); the leading correction term $r_{{\rm Origin},\xi}(s)$ with interior rescaling (red solid line), exterior rescaling (thin black line) and no rescaling (dashed green line).}\label{fig:3}
\end{figure}

\begin{figure}[tbp]
\hspace*{-0.5cm}
\includegraphics[width=0.55\textwidth]{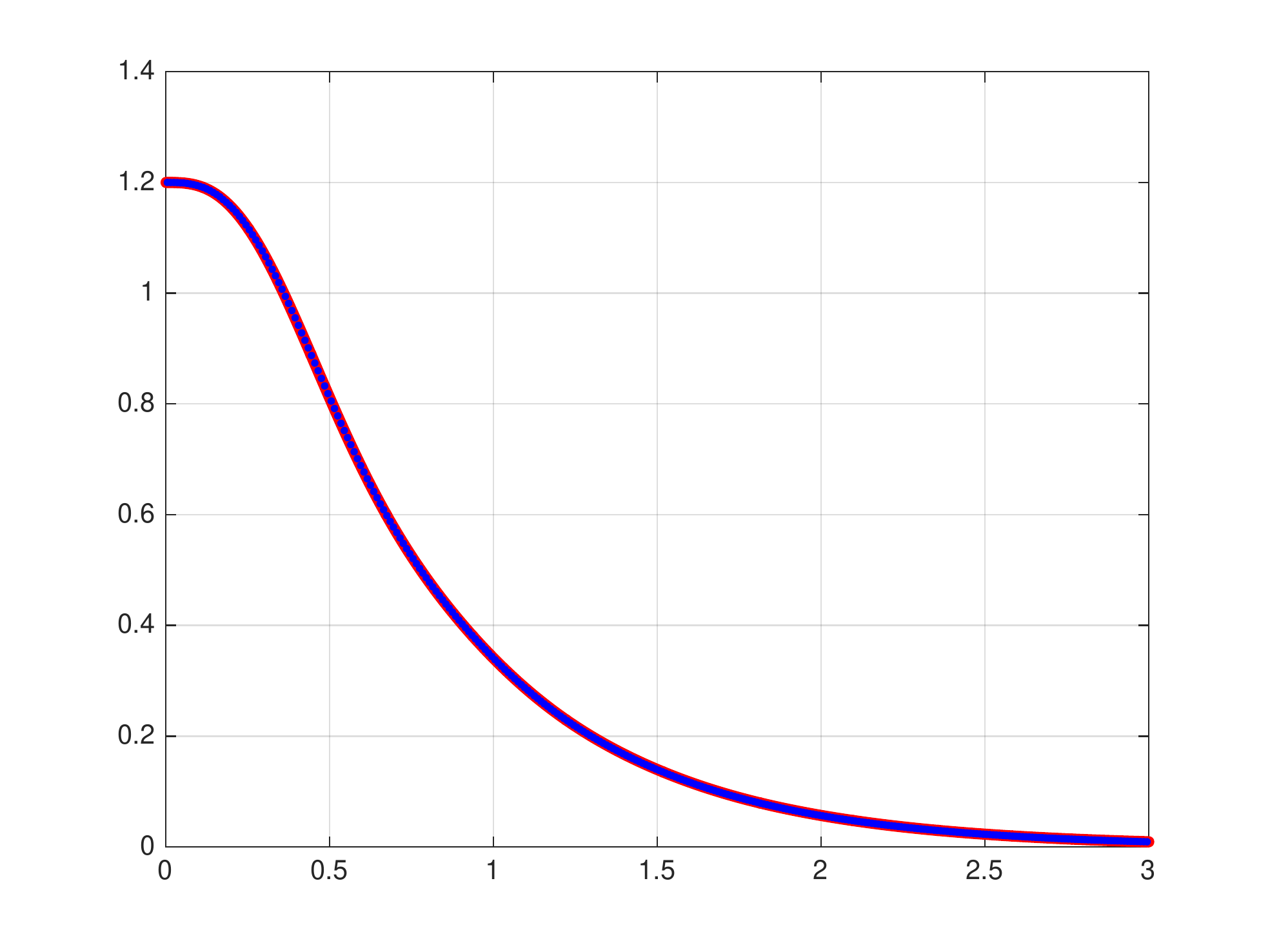}\hspace*{-0.5cm}
\includegraphics[width=0.55\textwidth]{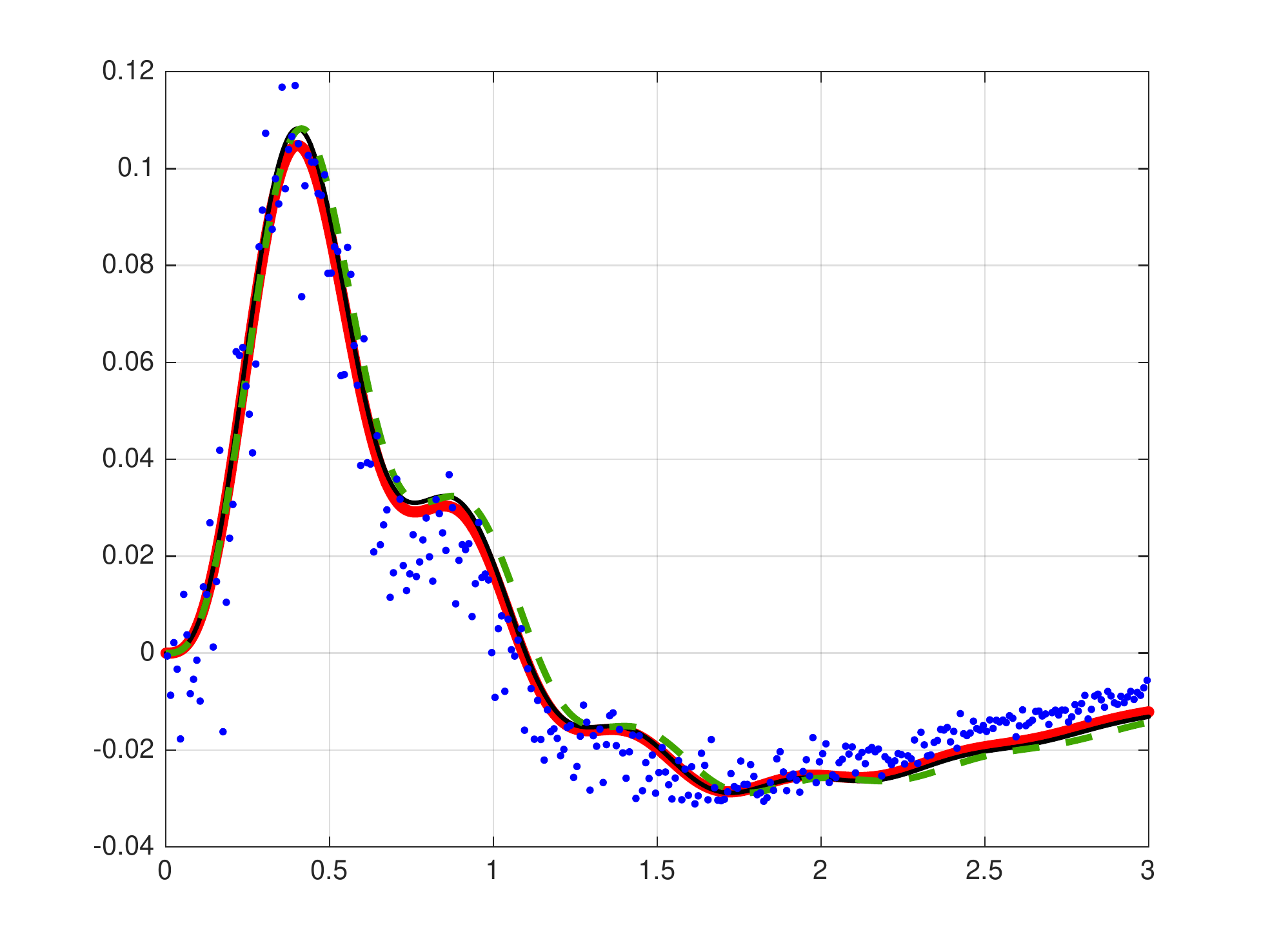}
\caption{As in Figure \ref{fig:3} but with thinning $\xi=0.6$.}\label{fig:4}
\end{figure}

\subsection{Interior rescaling of the leading correction terms}

The numerical data, physical arguments and mathematical
conjectures surrounding the fluctuations of the Riemann
zeros indicate that their statistics are asymptotically well
approximated, to more than just leading order, by some correlation kernel
$K_{\rm RZ}^N$. Consequently, to justify the mechanism of the exterior rescaling
introduced in the last paragraph, \cite[p.~10748]{BBLM06} suggested to pull-back the absorption of the $\O(\bar\rho^{-3})$ term from the two-point function into the expansion of such an assumed kernel.  
In the present paper we will use such an expansion to improve systematically the ${\rm CUE}_N$ fit of the fluctuation properties
of the Riemann zeros statistics from $\O(N^{-3})$ to $\O(N^{-4})$. 

By modifying the expansion~\eqref{eq:KNexpansion} of the CUE correlation
kernel $K^N$ to include an $\O(N^{-3})$ term,
\begin{equation}\label{eq:kernelexpansion}
K_{\rm RZ}^N(x,y) =  K(x,y) +  L(x,y) N^{-2} + M(x,y) N^{-3} + \O(N^{-4}),
\end{equation}
the task is thus to identify the unknown kernel $M(x,y)$ from the $\O(N^{-3})$ term in \eqref{eq:RRZ}. By translation invariance and symmetry
of the kernel $K_{\rm RZ}^N$ and by the normalization $K_{\rm RZ}^N(x,x)=1$ to mean spacing unity, we get the form
\[
M(x,y) = \mu(x-y)
\]
of a convolution kernel with a {\em symmetric} function $\mu$ satisfying $\mu(0)=0$. A brief calculation results in the
corresponding two-point function
\[
 R_2^{\rm RZ}(s) =  1 - \frac{\sin^2(\pi s)}{\pi^2 s^2} - \frac{\sin^2 (\pi s)}{3} N^{-2} -\frac{2\mu(s)\sin(\pi s)}{\pi s} N^{-3} + \O(N^{-4}).
 \]
Matching with \eqref{eq:RRZ} gives, cf. the first equality in \cite[Eq.~(28)]{BBLM06},
\[
\mu(s) = \eta \frac{\pi^2 s^2 \cos(\pi s)}{6},\qquad \eta = \frac{Q}{\Lambda\sqrt{3\Lambda}}.
\]
In the appendix of \cite{BBLM06} it is shown that this expansion of the kernel allows one to consistently reproduce a conjectural expansion of the three-point function (which was obtained from a calculation published much later in \cite{BK13}). This is further evidence that the kernel expansion \eqref{eq:kernelexpansion} induces, quite generally, the $O(N^{-3})$ correction terms in the determinantal point process fluctuation statistics.

Now, the correction kernel $M$ can be absorbed easily by an {\em interior} rescaling of the leading correction kernel $L$, namely 
by introducing
\begin{equation}\label{eq:Lbar}
 L_{\rm RZ}(x,y) = \frac{\pi(x-y)\sin(\pi\bar\alpha(x-y))}{6},
\end{equation}
which expands as $L_{\rm RZ}(x,y) = L(x,y) + M(x,y) N^{-1} + \O(N^{-2})$,
where\footnote{Note that \cite[Eq.~(28)]{BBLM06} contains a miscalculation by claiming that $\bar\alpha$ would be the same as $\alpha$.}
\begin{equation}\label{eq:alphabar}
  \bar \alpha = 1 + \eta N^{-1} = 1 + \frac{Q}{\pi\Lambda\bar\rho} =
1+ \frac{2Q}{\Lambda \log(E/2\pi)} = 2\alpha -1.
\end{equation}
To summarize, we have
\begin{equation}\label{eq:KRZ}
K_{\rm RZ}^N(x,y) =  K(x,y) +  L_{\rm RZ}(x,y) N^{-2} + \O(N^{-4})
\end{equation}
and thus get the following recipe to go from the leading correction term of a ${\rm CUE}_N$ fluctuation statistics to the
corresponding one of the Riemann zeros: simply replace the leading correction kernel $L$ by its interior rescaling $L_{\rm RZ}$.

In our experiments, the most pronounced difference between interior and exterior rescaling  can be observed in the statistics shown in
the right panel of Fig.~\ref{fig:4}. 

\subsection{The Odlyzko data set} The largest data set of Riemann zeros currently provided by Odlyzko, which was announced already in \cite{Od01},
consists of the $1\,041\,719\,075$ consecutive zeros starting with zero number
\[
10^{23} + 985\,531\,550.
\]
The first of them has height
\[
13\,066\,434\,408\,793\,621\,120\,027.39614\,65854\ldots,
\]
while the last one has height
\[
13\,066\,434\,408\,793\,754\,462\,591.63384\,74590\ldots.
\]
Because of the logarithmic dependence on the height, even within that large data set the smooth density $\bar\rho$, the effective dimension $N$ and the scaling parameters $\alpha$, $\bar\alpha$ 
remain essentially {\em constant} to 15 digits precision,
\begin{align*}
\bar\rho &= 7.81235\,22019\,1727\ldots,\\
N &= 11.29759\,09009\,547\ldots,\\
\alpha &= 1.02999\,00807\,6719\ldots,\\
\bar\alpha &= 1.05998\,01615\,3438\ldots.
\end{align*}

From this data set we extracted various spacing statistics,\footnote{The $0$-th next neighbour spacing can be found in \cite[Figs. 1--2]{Od01} (no thinning and no theoretical
prediction of the leading correction term), in \cite[Figs. 1--3]{BBLM06} (no thinning, exterior rescaling only) and
in \cite[Figs. 9--10]{FM15} (exterior rescaling only). Since in this case there is no difference
visible between interior and exterior rescaling we refrained from showing the results once more.} with ($\xi=0.6$) and without ($\xi=1$) thinning, and compared them up to the leading correction
with the theoretical prediction in terms of operator determinants $\det(\II-\KK)$ and their corrections $\Om{\KK}{\LL}$
as obtained in Section~\ref{sect:CUE} for the CUE. In Figs.~\ref{fig:1}--\ref{fig:8} the Riemann zero data are shown as blue dots, the leading order
term as a red solid line, and the correction terms 
\begin{itemize}
\item with interior rescaling $L\to L_{\rm RZ}$ as a red solid line,
\item with exterior rescaling $s\to \alpha s$ as a thin black line,
\item without rescaling as a green dashed line.
\end{itemize}

In particular, the following statistics of the Riemann zeros are shown to be in excellent agreement
up to the leading correction term with (interior) rescaling:

\medskip

{\begin{center}
\begin{tabular}{ccccc}\hline
spacing & density & CUE case$\phantom{\Big|}$ & no thinning & thinning $\xi=0.6$\\\hline 
$1$-st next neighbour & $p^{\rm RZ}_{\xi}(1;s)$ & Section~\ref{S2.3}$\phantom{\Big|}$ & Fig.~\ref{fig:1} & Fig.~\ref{fig:2}\\
random origin & $p^{\rm RZ}_{{\rm Origin},\xi}(s)$ & Section~\ref{sub:origin}$\phantom{\Big|}$ & Fig.~\ref{fig:3} & Fig.~\ref{fig:4}\\
nearest neighbour & $p^{\rm RZ}_{{\rm nn},\xi}(s)$ & Section~\ref{sub:nn}$\phantom{\Big|}$ & Fig.~\ref{fig:7} & Fig.~\ref{fig:8}\\\hline
\end{tabular}
\end{center}}

\medskip
Here we denote by $p^{\rm RZ}_{\;\ldots}$ those probability densities for the Riemann zeros that are defined analogously to 
the CUE case $p^{\rm CUE}_{\;\;\ldots}$. 
\begin{figure}
\hspace*{-0.5cm}
\includegraphics[width=0.55\textwidth]{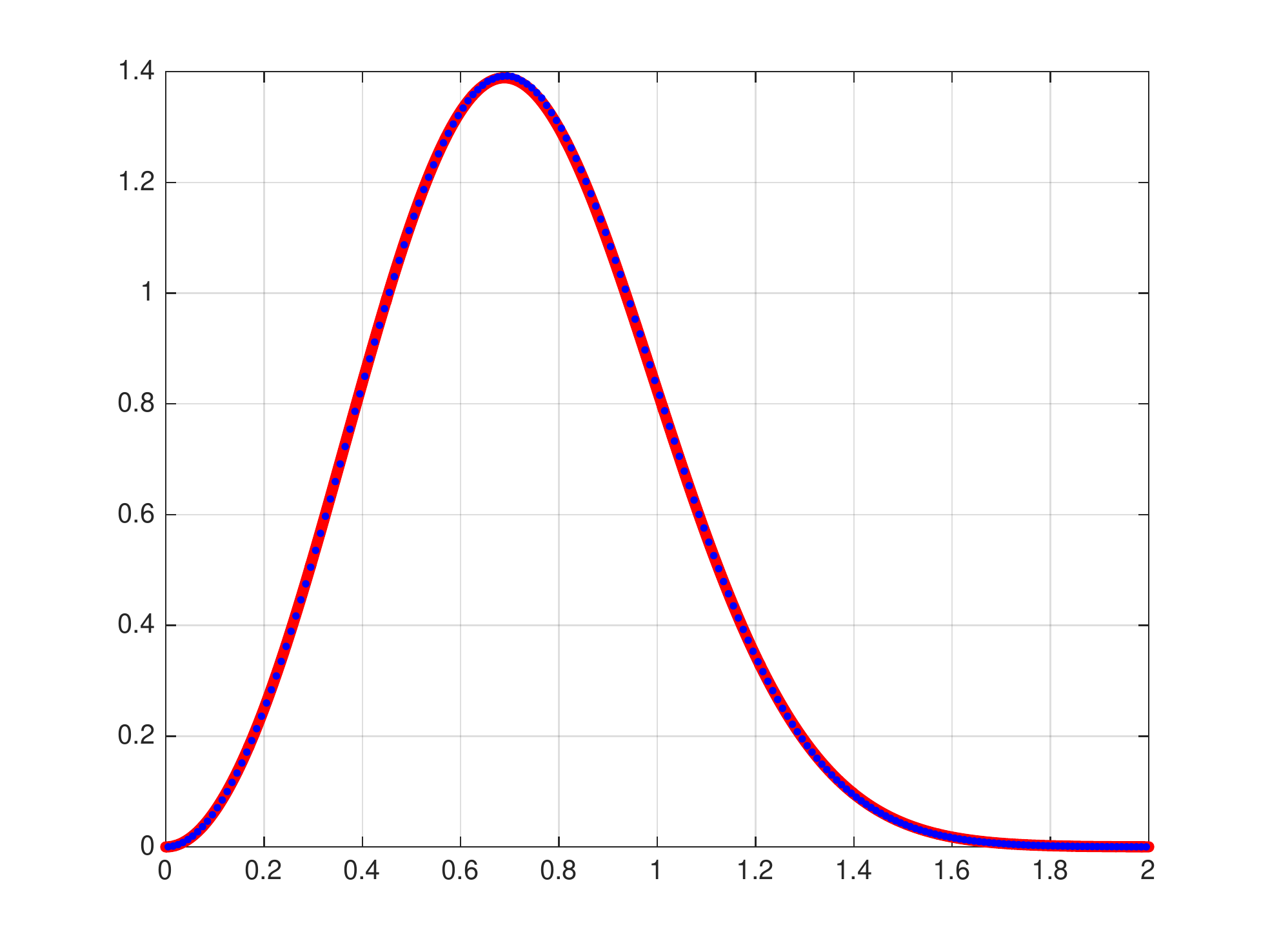}\hspace*{-0.5cm}
\includegraphics[width=0.55\textwidth]{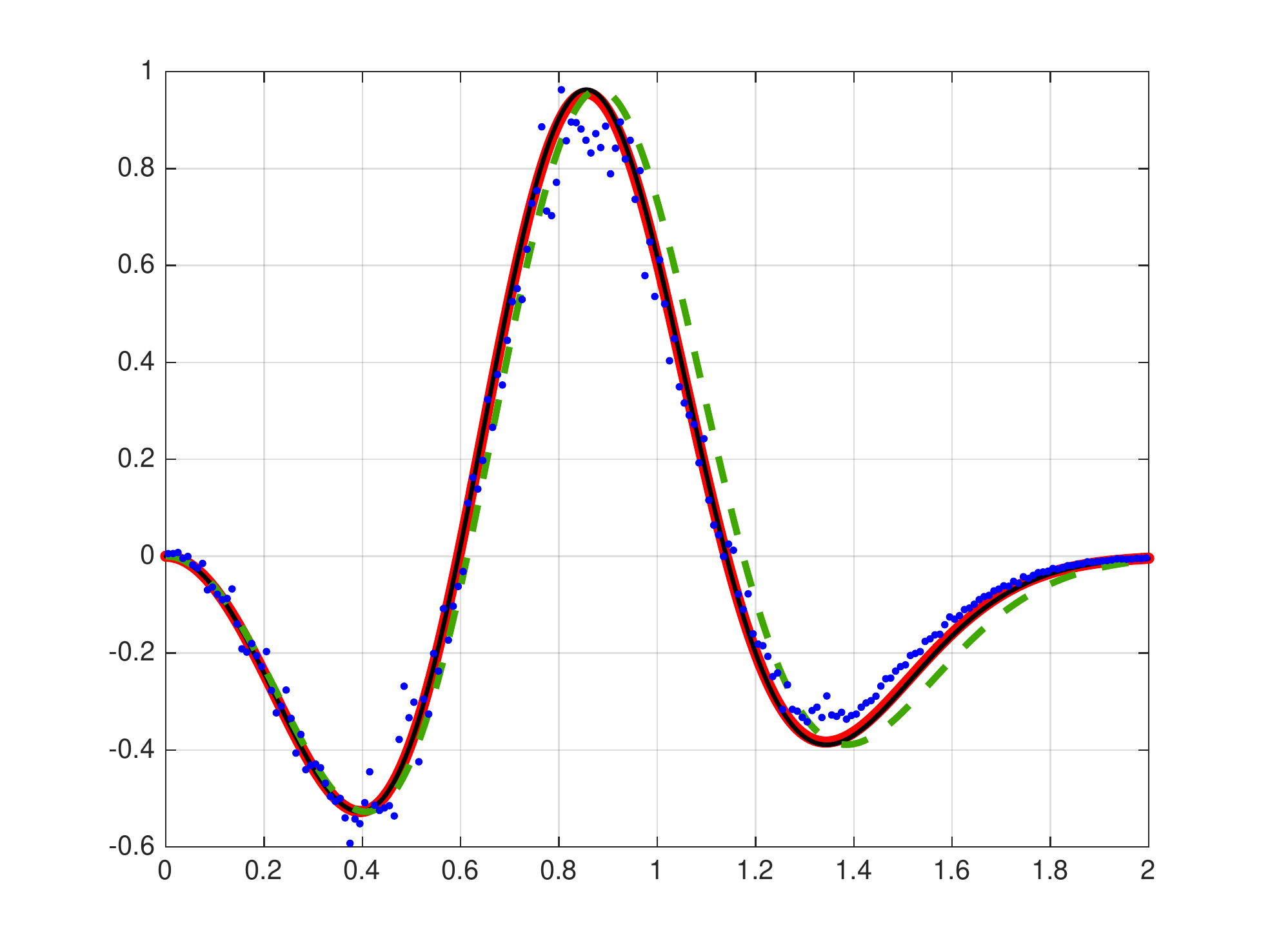}
\caption{Nearest neighbour spacing: Riemann zeros data vs. random-matrix based prediction with effective dimension $N=\log(E/2\pi)/\sqrt{12\Lambda}\approx11.3$, no thinning ($\xi=1$). Left panel: a histogram of the Odlyzko data set with bin size $0.01$ (blue dots); 
the large $N$ limit $p_{{\rm nn},\xi}(s)$ (red solid line). Right panel: the Riemann zero data minus $p_{{\rm nn},\xi}(s)$  scaled
by $N^2$ (blue dots); the leading correction term $r_{{\rm nn},\xi}(s)$ with interior rescaling (red solid line), exterior rescaling (thin black line) and no rescaling (dashed green line).}\label{fig:7}
\end{figure}

\begin{figure}
\hspace*{-0.5cm}
\includegraphics[width=0.55\textwidth]{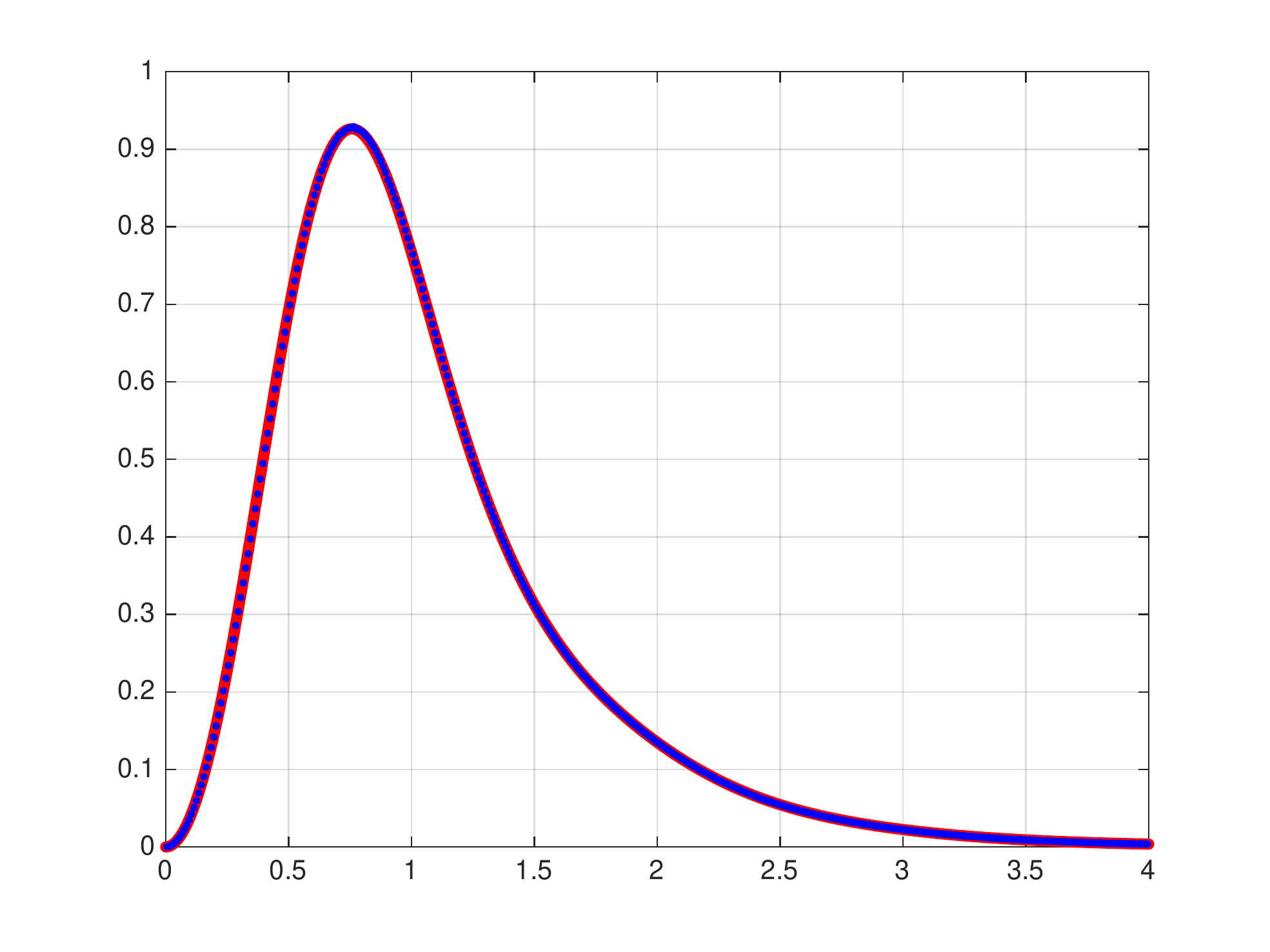}\hspace*{-0.5cm}
\includegraphics[width=0.55\textwidth]{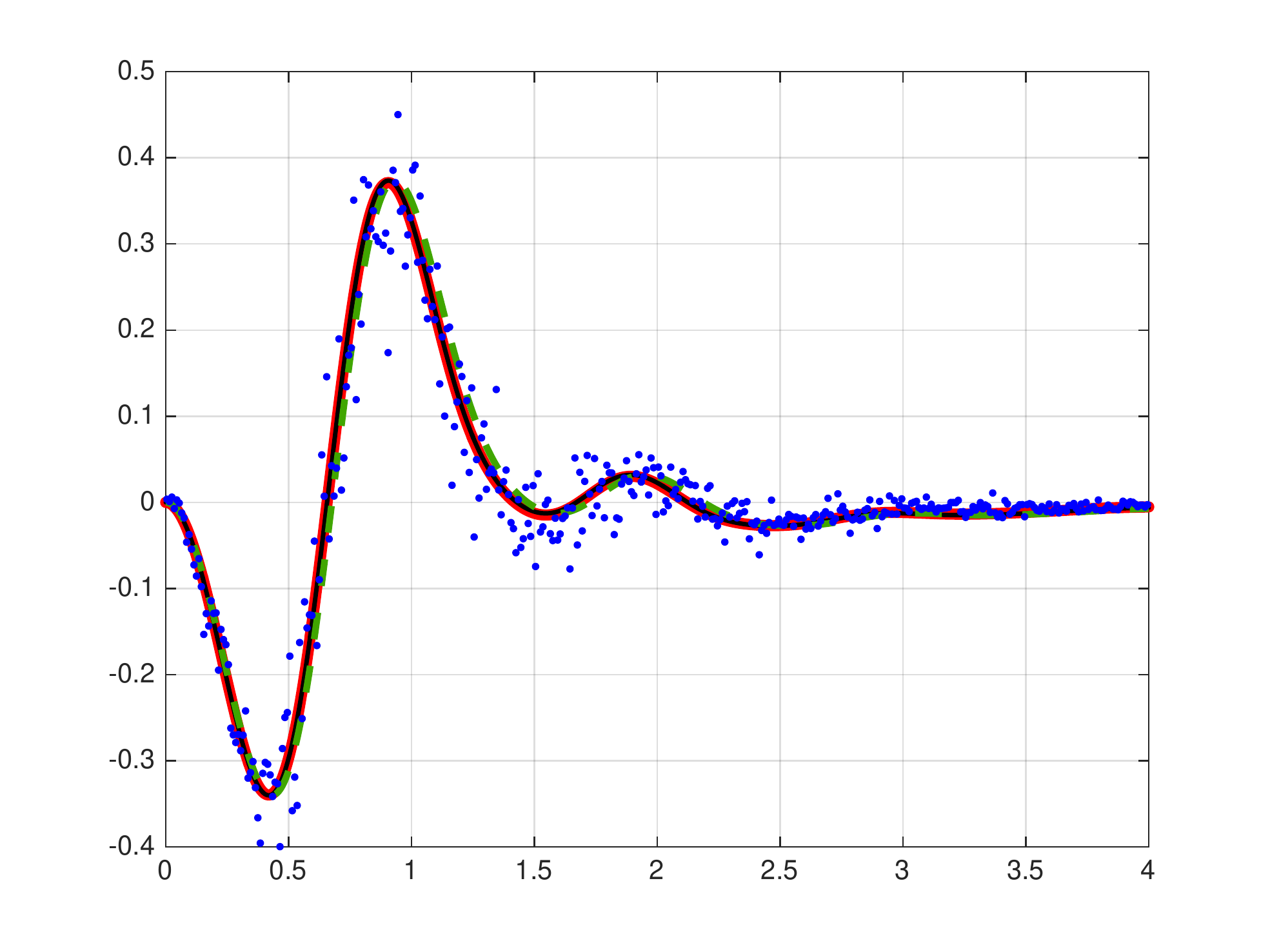}
\caption{As in Figure \ref{fig:7} but with thinning $\xi=0.6$.}\label{fig:8}
\end{figure}

\section{Spacing distributions for the COE and CSE}\label{sect:COECSE}
\subsection{Generating functions}
As already mentioned, the fact that the eigenvalues of COE matrices form a Pfaffian point process tells us that
the counterpart to the formula (\ref{RK}) involves a Pfaffian of a $2 \times 2$ anti-symmetric kernel function.
Substituting this in (\ref{ER}), the Pfaffian analogue of the identity implying (\ref{RK1}) (see e.g.~\cite[Eq.~(6.32)]{Fo10})
allows for a Pfaffian analogue in the case of the COE to be given. However, the resulting formula is not of the same
practical utility as (\ref{RK1}) since the $2 \times 2$ kernel function is not analytic, which in turn means that the
numerical quadrature methods of Section~\ref{sub:numerical} have poor convergence properties.

Fortunately, there is a second option. This presents itself due to fundamental inter-relations between
gap probabilities for the COE and CUE. These also involve the gap probabilities for Haar distributed
orthogonal matrices from the classical groups O${}^+(N)$ and O${}^-(N)$.
In reference to the latter, eigenvalues on the real axis and thus corresponding to $\theta = 0$ or $\pi$ are
to be disregarded, and only the eigenvalues with angles in $(0,\pi)$ are to be considered (the remaining eigenvalues
occur at the negative of these angles, i.e.~the complex conjugates). The first inter-relation of interest is
between the generating functions for the gap probabilities in the CUE, and in
Haar distributed real orthogonal matrices \cite{BR01a,Fo06}, \cite[Eq.~(8.127)]{Fo10}
\begin{equation}\label{W1}
{\mathcal E}^{\rm CUE}((-\theta,\theta);z) =
{\mathcal  E}^{\O^-(2\lfloor (N+1)/2 \rfloor + 1)}((0,\theta);z)
{\mathcal E}^{\O^+(2\lfloor N/2 \rfloor + 1)}((0,\theta);z).
\end{equation}
With ME denoting a particular matrix ensemble,
we have used the symbol ${\mathcal E}^{\rm ME}$ instead of $E^{\rm ME}$
to indicate that in this expression no scaling of the eigen-angles has been
imposed; this convention will be followed below. In fact we only use the scaled quantity for ensembles ${\rm ME} = {\rm C}\beta{\rm E}_N$ which are translationally invariant with uniform density $N/2 \pi$. The unscaled and scaled quantities are then related by
$$
{\mathcal E}^{ {\rm C} \beta{\rm E}_N} ((-2 \pi s/N, 2 \pi s/N);z) = E^{ {\rm C}\beta{\rm E}_N}(2s;z).
$$

To present the next inter-relation, introduce the generating functions
\begin{equation}\label{W2}
E^{{\rm C O E}, \pm}((-\theta,\theta);z)  := \sum_{n=0}^N (1 - z)^n \Big (
E^{{\rm C O E}}(2n,(0,\theta)) +  E^{{\rm C O E}}(2n \pm 1,(0,\theta)) \Big ).
\end{equation}
With COE$\, \cup \,$COE denoting the point process of $2N$ eigenvalues on the circle
that results by superimposing the eigenvalue sequences of two independent
COE matrices, a result conjectured by Dyson \cite{Dy62a} and proved by Gunson \cite{Gu62}
gives that
\begin{equation}\label{W3}
{\rm alt} \left({\rm COE}\, \cup \,{\rm COE}\right) = {\rm CUE},
\end{equation}
where the operation ``alt" refers to the operation of observing every second eigenvalue only.
As a consequence of this, one has that  \cite{Dy62a}, \cite{Me92}
\begin{equation}\label{W4}
E^{\rm CUE}((-\theta,\theta);z) = E^{{\rm C O E}, -}((-\theta,\theta);z)  E^{{\rm C O E}, +}((-\theta,\theta);z).
\end{equation}

According to (\ref{W1}), it follows from (\ref{W4}) that the gap probabilities for the COE are related to those
for   O${}^+(n)$ and O${}^-(n)$ with $n$ suitably chosen, but this alone does not determine 
the former. To be able to do this, additional inter-relationships between generating functions are required. 
In the case $N \mapsto 2N$ and thus $N$ even, the additional inter-relationships have been given
in \cite{Fo06} according to the generating function identity
\begin{equation}\label{W5}
{\mathcal E}^{{\rm C O E}, \pm}((-\theta,\theta);z) \Big |_{N \mapsto 2N} =  
{\mathcal E}^{\O^\pm(2N + 1)}((0,\theta);z).
\end{equation}
Note that (\ref{W5}) substituted in (\ref{W4}) is consistent with (\ref{W1}) in the case $N$ even.
The identity (\ref{W5}) has very recently \cite{BF15}
been shown to be a corollary of the identities between eigenvalue distributions
\begin{equation}\label{W6}
{\rm even} \, |{\rm COE}_{2N} | = \O^+(2N+1), \qquad
{\rm odd} \, |{\rm COE}_{2N} | = \O^-(2N+1).
\end{equation}
Here the notation ${\rm COE}_{2N}$ refers to the eigenvalue distribution of $2N \times 2N$ COE
matrices, while $|{\rm COE}_{2N}|$ refers to the distribution in the circumstance that the eigenvalues
with angles $-\pi < \theta < 0$ are reflected in the real axis by $\theta \mapsto - \theta$, and thus all
eigenvalues have angles between 0 and $\pi$. The operation even (odd) refers to observing only
those eigenvalues that occur an even (odd) number of places from $\theta = 0$, reading anti-clockwise.

Substituting (\ref{W5}) in (\ref{W2}), and setting $\bar{z} = 2 z - z^2$ so that $1 - \bar{z }= (1 - z)^2$
 we obtain,  after some minor manipulation,  a known closed formula for the  gap probabilities of ${\rm COE}_{2N}$ in terms of the gap probabilities for $\O^\pm(2N+1)$ \cite[Eq.~(8.150)]{Fo10}
\begin{equation}\label{W7} 
{\mathcal E}^{{\rm C O E}_{2N} }((-\theta,\theta);z)  =
{(1 - z) \mathcal E^{\O^+(2N+1)}((0,\theta);\bar{z}) +  \mathcal E^{\O^-(2N+1)}((0,\theta);\bar{z}) \over
2 - z}.
\end{equation}

The analogue of (\ref{W6}) for ${\rm COE}_{2N-1}$ allows us to deduce the analogue of
(\ref{W7}). The required formulas were not known until \cite{BF15} and consequently this is a new result.

\begin{prop}
We have
\begin{equation}\label{W8} 
\mathcal E^{{\rm C O E}_{2N-1}}((-\theta,\theta);z) =
{(1 - z) \mathcal  E^{\O^-(2N)}((0,\theta);\bar{z}) +  \mathcal  E^{\O^+(2N)}((0,\theta);\bar{z}) \over
2 - z}.
\end{equation}
\end{prop}

\begin{proof}
We read off from \cite[Th.~7.1]{BF15} that
\begin{equation}\label{W6a}
{\rm even} \, |{\rm COE}_{2N-1} | = \O^-(2N), \qquad
{\rm odd} \, |{\rm COE}_{2N-1} | = \O^+(2N).
\end{equation}
Substituting in (\ref{W2}) and manipulating as in the derivation of (\ref{W7}) gives (\ref{W8}).
\end{proof}

With $\nu = (-)^N$ the cases \eqref{W7} and \eqref{W8} can be combined into a single formula that holds for both parities
of $N$:
\begin{equation}\label{W8i} 
\mathcal E^{{\rm C O E}_{N}}((-\theta,\theta);z)
 = {(1 -  z) \mathcal  E^{\O^{\nu}(N+1)}((0,\theta);\bar z) +  \mathcal  E^{\O^{-\nu}(N+1)}((0,\theta);\bar z) \over
2 -  z}.
\end{equation}
The eigenvalues for $\O^\pm(N+1)$ in $(0,\pi)$ form a determinantal point process with kernel (see e.g.~\cite[Prop.~5.5.3]{Fo10})
\begin{equation}\label{eq:OpmKernel}
 {1 \over 2 \pi} \bigg ( {\sin N (x - y)/2 \over \sin (x - y)/2} \mp \nu {\sin N (x + y)/2 \over \sin (x + y)/2} \bigg )
 = \frac{N}{2\pi} K^{N,\mp\nu}(x N/2\pi, y N/2\pi),
\end{equation}
 where 
\begin{equation}\label{eq:Kpm}
 K^{N,\pm}(x,y) = K^N(x,y) \pm K^N(x,-y).
 \end{equation}
Scaling the eigen-angles of ${\rm COE}_N$ to have unit mean spacing by $\theta =2\pi s/N$ thus yields the following result:
 
\begin{corollary}\label{cor:COE} Let $\KK^{N,\pm}_s$ denote the integral operator on $(0,s)$ with kernel \eqref{eq:Kpm}
and denote $\bar{z} = 2 z - z^2$. We have
\begin{equation}\label{W8ii}
E^{{\rm C O E}_{N}}((-s,s);z) = \frac{(1-z)\det(I-\bar z \KK^{N,-}_s)+\det(I-\bar z \KK^{N,+}_s)}{2-z}.
\end{equation}
As in \eqref{Th} thinning is expressed by
$E_\xi^{{\rm C O E}_{N}}((-s,s);z) = E^{{\rm C O E}_{N}}((-s,s);\xi z).$ 
\end{corollary}

We will now turn our attention to spacing probabilities for the CSE. These follow from knowledge of the spacing
distributions for the COE. Thus one has the inter-relation \cite{MD63}
\begin{equation}\label{W6b}
{\rm alt} \, {\rm COE}_{2N} = {\rm CSE}_N,
\end{equation}
and it follows from this that
\begin{multline}\label{W6B}
\mathcal  E^{{\rm CSE}_N}(n;(-\theta,\theta)) \\
 = \mathcal  E^{{\rm COE}_{2N}}(2n;(-\theta,\theta)) + {1 \over 2} \Big (
\mathcal  E^{{\rm COE}_{2N}}(2n-1;(-\theta,\theta)) + \mathcal  E^{{\rm COE}_{2N}}(2n+1;(-\theta,\theta))  \Big ).
\end{multline}
As noted in \cite[Eq.~(8.158)]{Fo10}, recalling the definition (\ref{W2}), and making use too of~(\ref{W5}), this gives
\begin{equation}\label{W6c}
\mathcal  E^{{\rm CSE}_N}((-\theta,\theta);z)  = {1 \over 2} \Big ( \mathcal E^{\O^+(2N+1)}((0,\theta);z) + \mathcal  E^{\O^-(2N+1)}((0,\theta);z) \Big ).
\end{equation}
As discussed before \eqref{eq:OpmKernel}, the kernel of the determinantal point process formed by the eigenvalues
of $\O^\pm(2N+1)$ 
is given by
\[
\frac{N}{\pi} K^{2N,\mp}(x N/\pi,y N/\pi).
\]
Scaling the eigen-angles of ${\rm CSE}_N$ to have unit mean spacing by $\theta =2\pi s/N$ thus yields the following result,
where we use \eqref{Th} to express the presence of thinning:

\begin{prop}\label{cor:CSE}
 Let $\KK^{N,\pm}_s$ denote the integral operator on $(0,s)$ with kernel~\eqref{eq:Kpm}. We have
\begin{equation}\label{W6C}
E_\xi^{{\rm C S E}_{N}}((-s,s);z) = \frac{1}{2}\left(\det(I- \xi z \KK^{2N,-}_{2s})+\det(I- 
\xi z \KK^{2N,+}_{2s})\right).
\end{equation}
\end{prop}

  \begin{figure}
\hspace*{-0.5cm}
\includegraphics[width=0.55\textwidth]{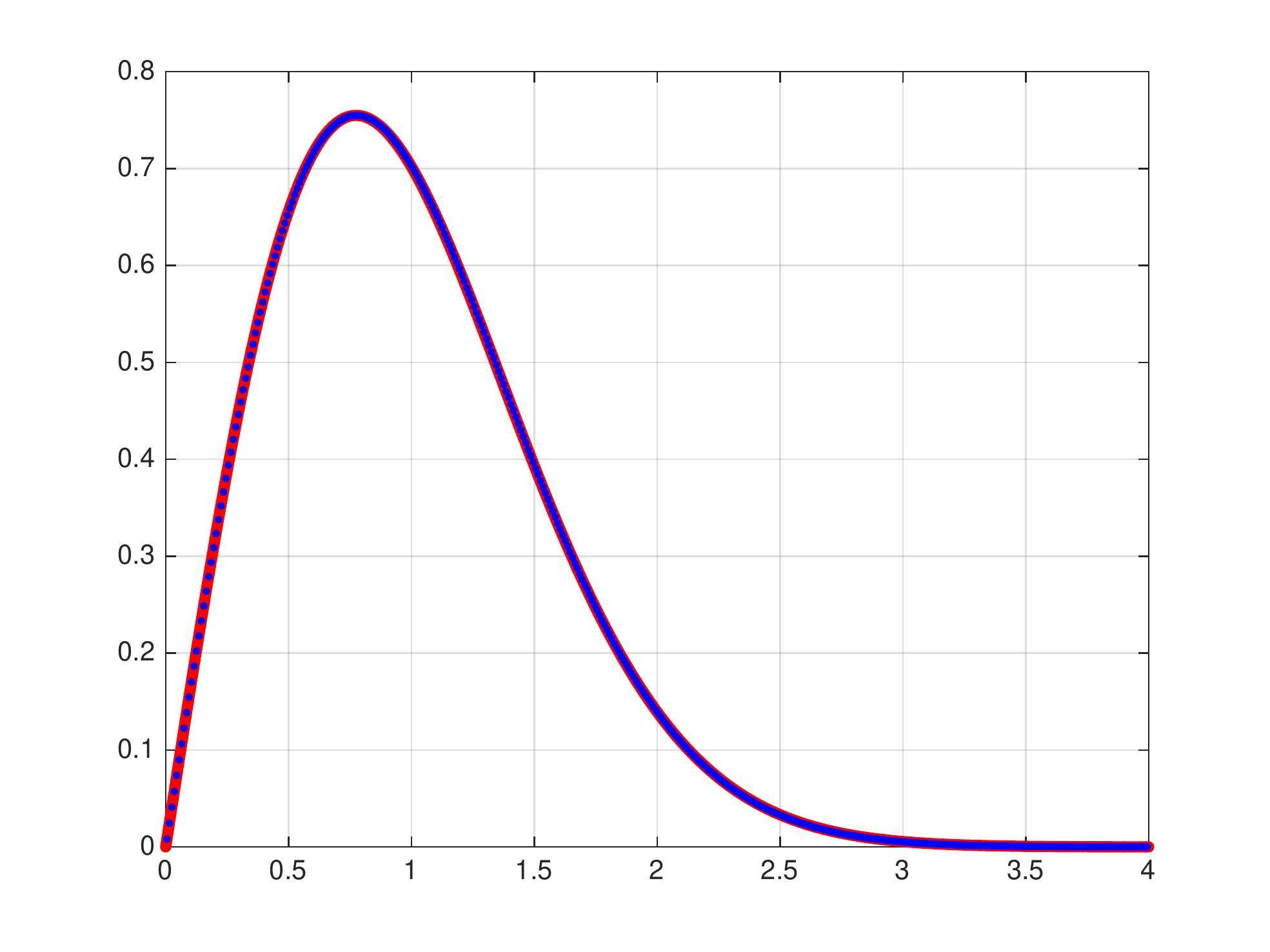}\hspace*{-0.5cm}
\includegraphics[width=0.55\textwidth]{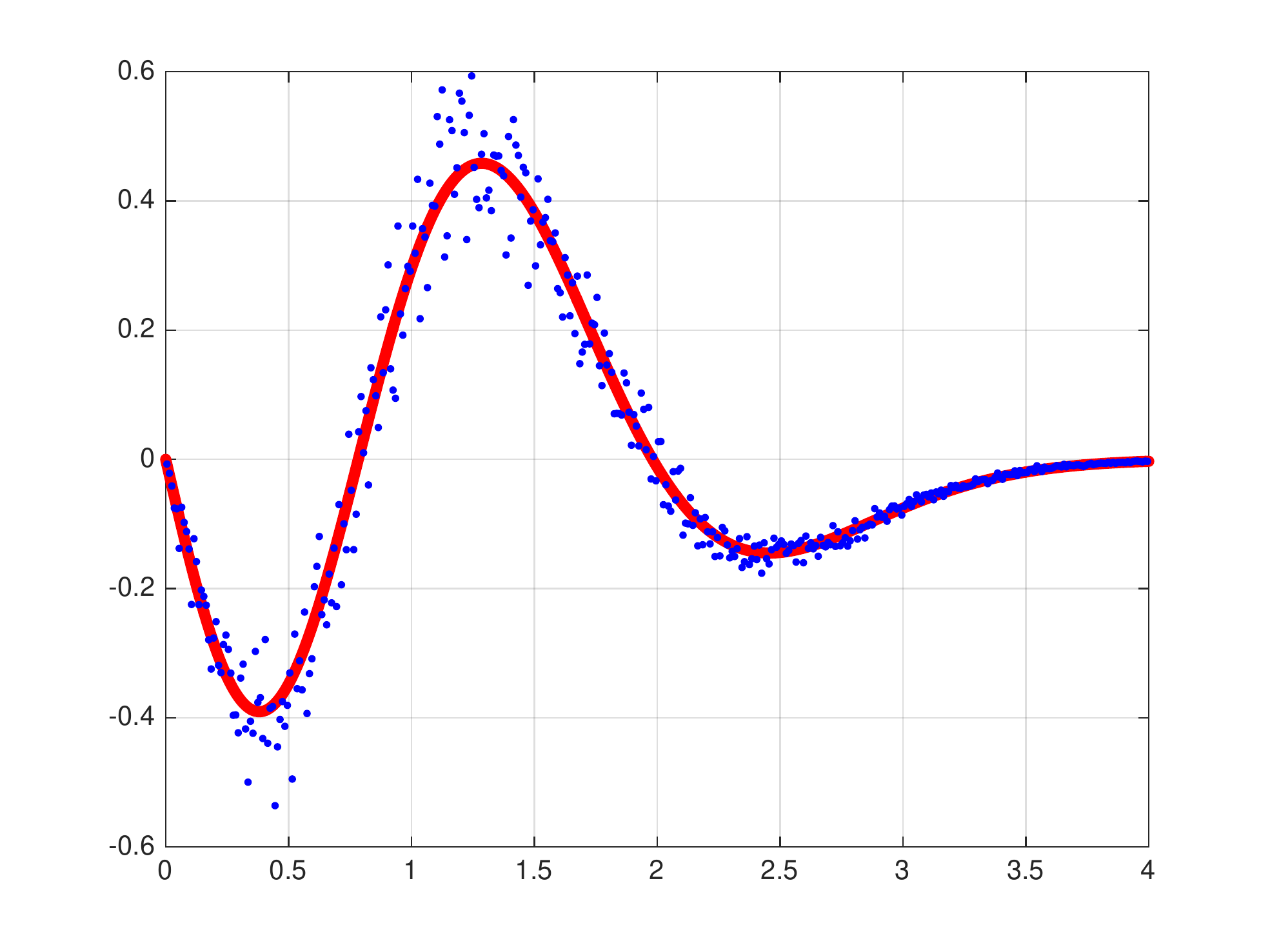}
\caption{\label{F8} $0$-th next neighbour spacing: simulation vs. formulae from theory for finite size COE, no thinning ($\xi=1$). Left panel: a histogram of empirical data from  ${\rm COE}_{N}$ with $N=20$ scaled to unit mean spacing, computed using a bin size of 0.01 and $10^8$ samples (blue dots); the large $N$ limit $p_{1,\xi}(0;s)$ (red solid line). Right panel: the simulation data minus $p_{1,\xi}(0;s)$  scaled
by $N^2$ (blue dots); the leading correction term $r_{1,\xi}(0,s)$ (red solid line).}
\end{figure}

 \begin{figure}
\hspace*{-0.5cm}
\includegraphics[width=0.55\textwidth]{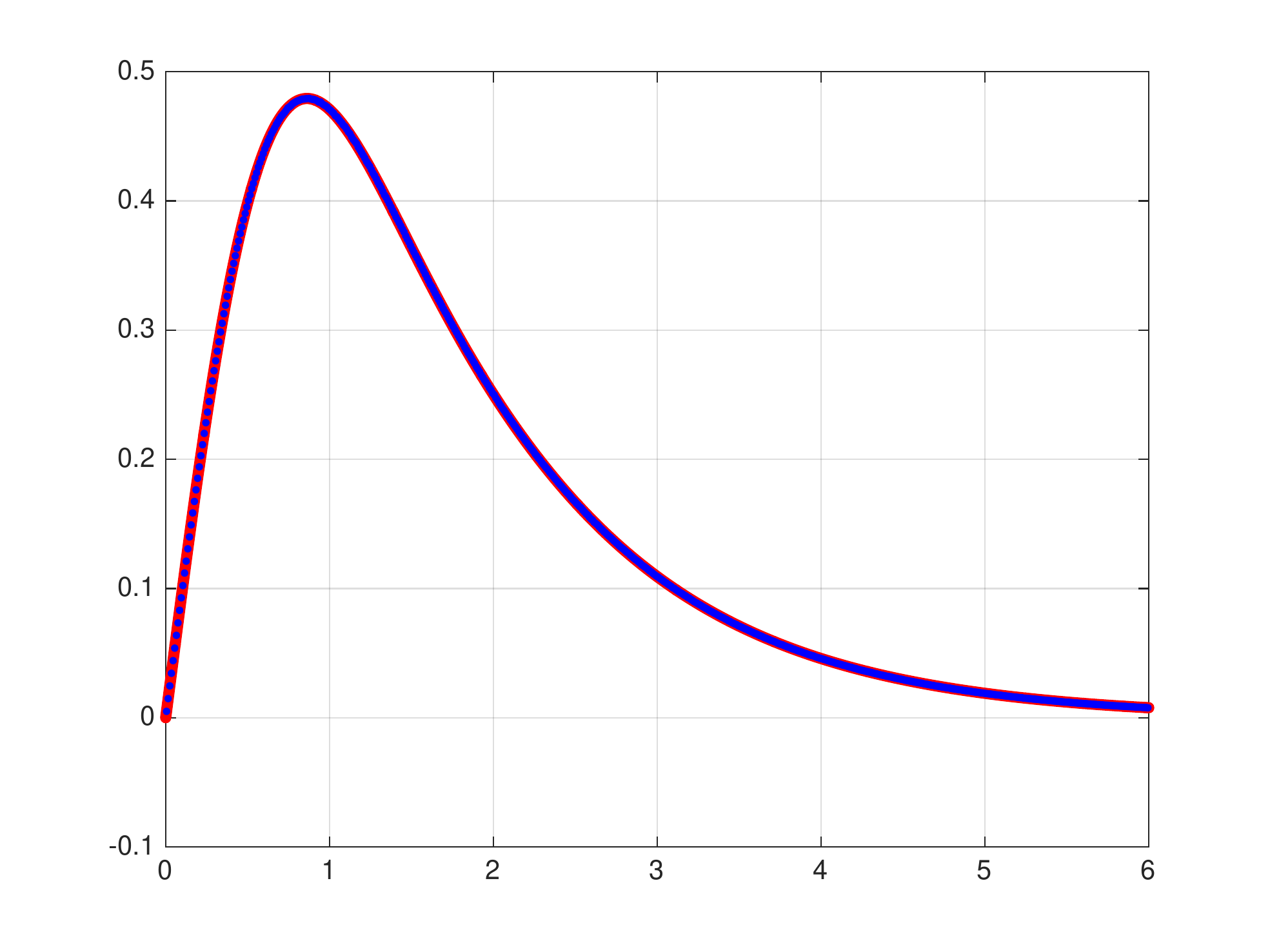}\hspace*{-0.5cm}
\includegraphics[width=0.55\textwidth]{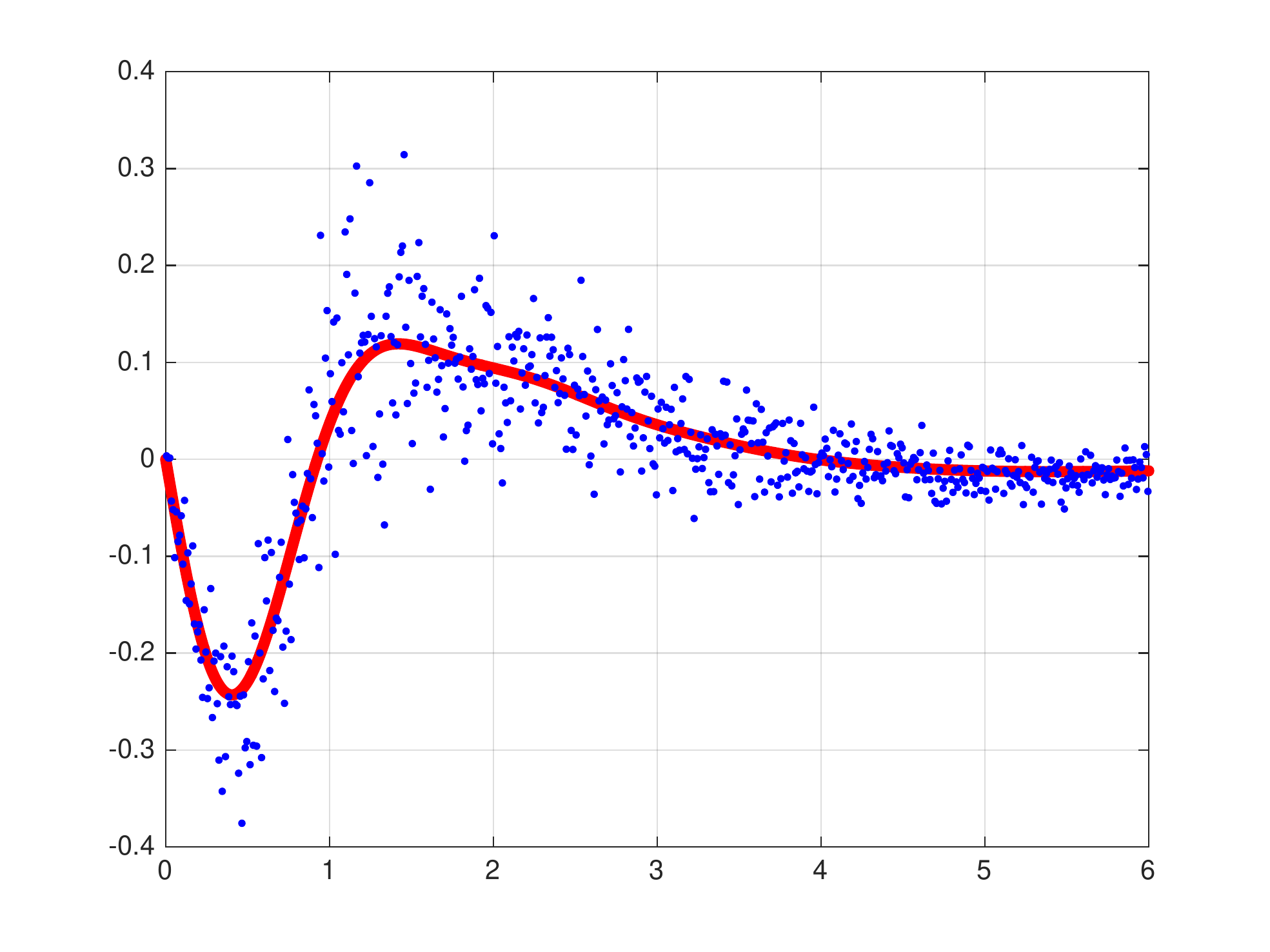}\caption{\label{F8a}As in Figure \ref{F8} but with thinning $\xi=0.6$. }
\end{figure}

\subsection{Expansion of spacing distributions} 

Making use of the second equation in (\ref{Th2}) we get the spacing 
distribution
\begin{equation}\label{Sa1}
 p_\xi^{{\rm COE}}(0;s)  = {1 \over  \xi} {d^2 \over d s^2} E_\xi^{{\rm COE}_N}(0;(-s/2,s/2)),
\end{equation}
from which knowledge of the large $N$ expansion of 
the determinants in Corollary~\ref{cor:COE} will allow us to determine the terms in the corresponding expansion of the spacing.

\begin{prop}\label{prop:COE} Let $\KK_s^\pm$ and $\LL_s^\pm$ denote the integral operators on $(0,s)$ with kernels $K(x,y) \pm K(x,-y)$
and $L(x,y)\pm L(x,-y)$ respectively.
 We have
\begin{equation}\label{Sa}
p^{{\rm COE}}_\xi(0;s)  = p_{1, \xi}(0;s) + {1 \over N^2} r_{1, \xi} (0;s) + \O\Big ( {1 \over N^4} \Big ),
\end{equation}
where, denoting $\bar{\xi} = 2 \xi - \xi^2$,
\begin{equation}\label{M1}
 p_{1,\xi}(0;s) = {1 \over \bar \xi}  {d^2 \over d s^2} \left( (1 - \xi)  
 \det( \II - \bar\xi \KK_{s/2}^-)  +  \det( \II - \bar\xi \KK_{s/2}^+)  \right)
\end{equation}
and 
\begin{equation}
r_{1,\xi}(0;s) = {d^2 \over d s^2} \left( (1 - \xi)  
 \Om{\bar\xi \KK_{s/2}^-}{\LL_{s/2}^-}  +  \Om{\bar\xi \KK_{s/2}^+}{\LL_{s/2}^+} \right)
 \end{equation}
 \end{prop}

  \begin{remark}
 The case $\xi = 1$ of (\ref{M1}) is due to Gaudin \cite{Ga61}.
 \end{remark}

As an illustration and test of the above results, we took $10^8$ samples of ${\rm COE}_N$  ($N=20$,  using the $\beta=1$ CMV sparse matrix model \cite{KN04}), and from this made an empirical determination of
the spacing distribution scaled to unit mean spacing. This was then subtracted from the large $N$ limit
$p_{1,\xi}(0;s)$ and the difference compared against $r_{1,\xi}(0;s)$. Both
$\xi =1$ (no thinning) and $\xi=0.6$ were considered; see Figures \ref{F8} and \ref{F8a}.

 The analogue of (\ref{Sa1}),
\begin{equation}\label{Sa16}
 p_\xi^{{\rm CSE}}(0;s)  = {1 \over  \xi} {d^2 \over d s^2} E_\xi^{{\rm CSE}_N}(0;(-s/2,s/2)),
\end{equation}
allows us to determine the first terms of the large $N$ expansion by expanding correspondingly the determinants in Proposition~\ref{cor:CSE}.

\begin{prop}\label{prop:CSE}
With $\KK_s^\pm$ and $\LL_s^\pm$ as in Proposition \ref{prop:COE} we have
\begin{equation}\label{Sa7}
 p_\xi^{{\rm CSE}}(0; s) = p_{4,\xi}(0;s) + {1 \over N^2} r_{4,\xi}(0;s) + \O\Big ( {1 \over N^4} \Big ),
\end{equation}
where
\begin{equation}\label{M2}
 p_{4,\xi}(0;s) = {1 \over 2 \xi }  {d^2 \over d s^2} \left (
 \det ( \II - \xi \KK_s^-)  +  \det ( \II - \xi  \KK_s^+)  \right )
\end{equation}
and
\begin{equation}
 r_{4,\xi}(0;s) = {1 \over 8}  {d^2 \over d s^2} \left (
 \Om{\xi \KK_s^-}{\LL_s^-}  +  \Om{\xi  \KK_s^+}{\LL_s^+}  \right ).
 \end{equation} 
 \end{prop}

   \begin{remark}
 The case $\xi = 1$ of (\ref{M2}) is due to Mehta and Dyson \cite{MD63}.
 \end{remark}

 \begin{figure}
\hspace*{-0.5cm}
\includegraphics[width=0.55\textwidth]{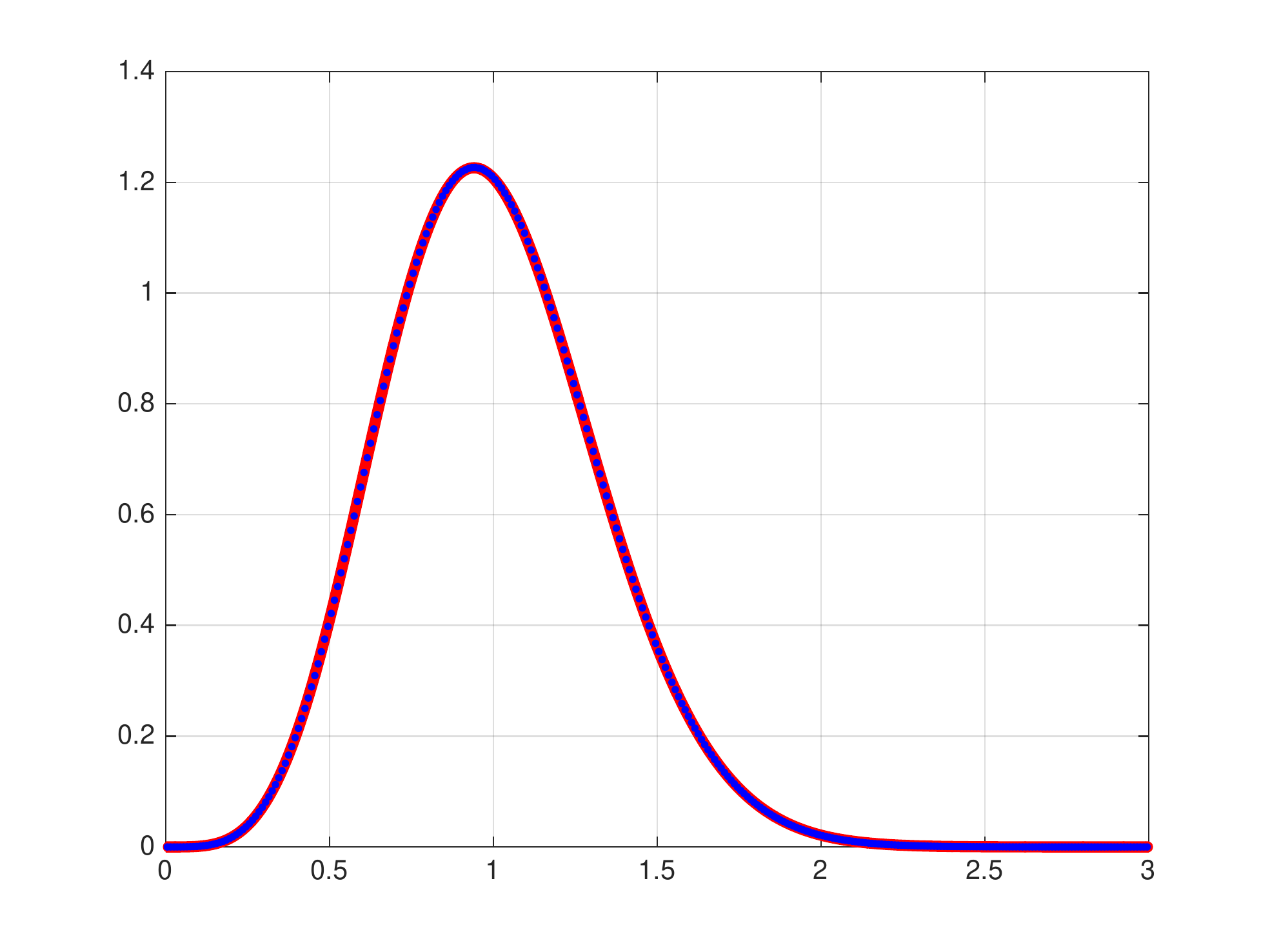}\hspace*{-0.5cm}
\includegraphics[width=0.55\textwidth]{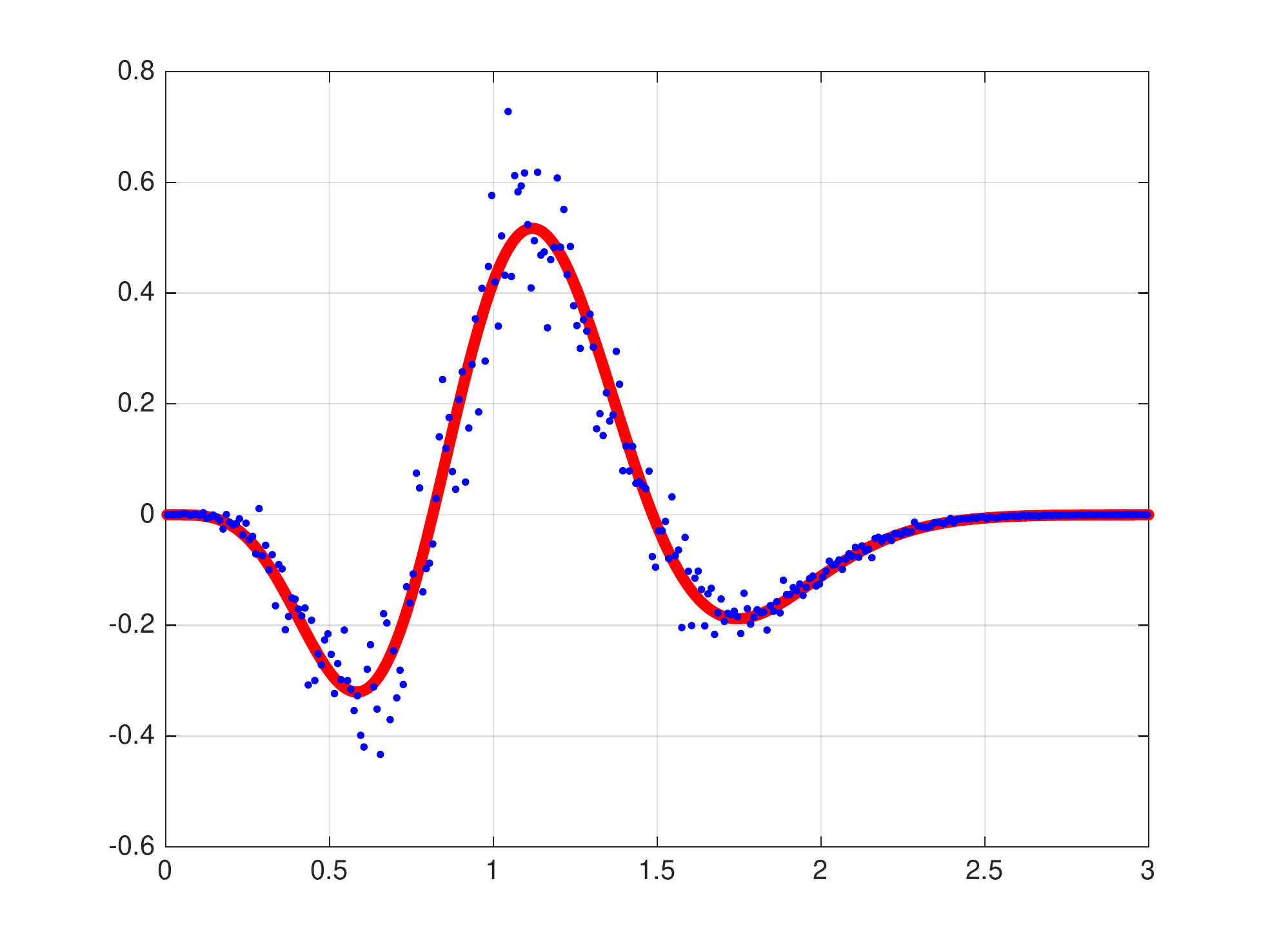}
\caption{\label{F9}$0$-th next neighbour spacing: simulation vs. formulae from theory for finite size CSE, no thinning ($\xi=1$). Left panel: a histogram of empirical data from ${\rm CSE}_{N}$ with $N=20$ scaled to unit mean spacing, computed using a bin size of $0.01$ and $10^8$ samples (blue dots); the large $N$ limit $p_{4,\xi}(0;s)$ (red solid line). Right panel: the simulation data minus $p_{4,\xi}(0;s)$  scaled
by $N^2$ (blue dots); the leading correction term $r_{4,\xi}(0,s)$ (red solid line).}
\end{figure}

\begin{figure}
\hspace*{-0.5cm}
\includegraphics[width=0.55\textwidth]{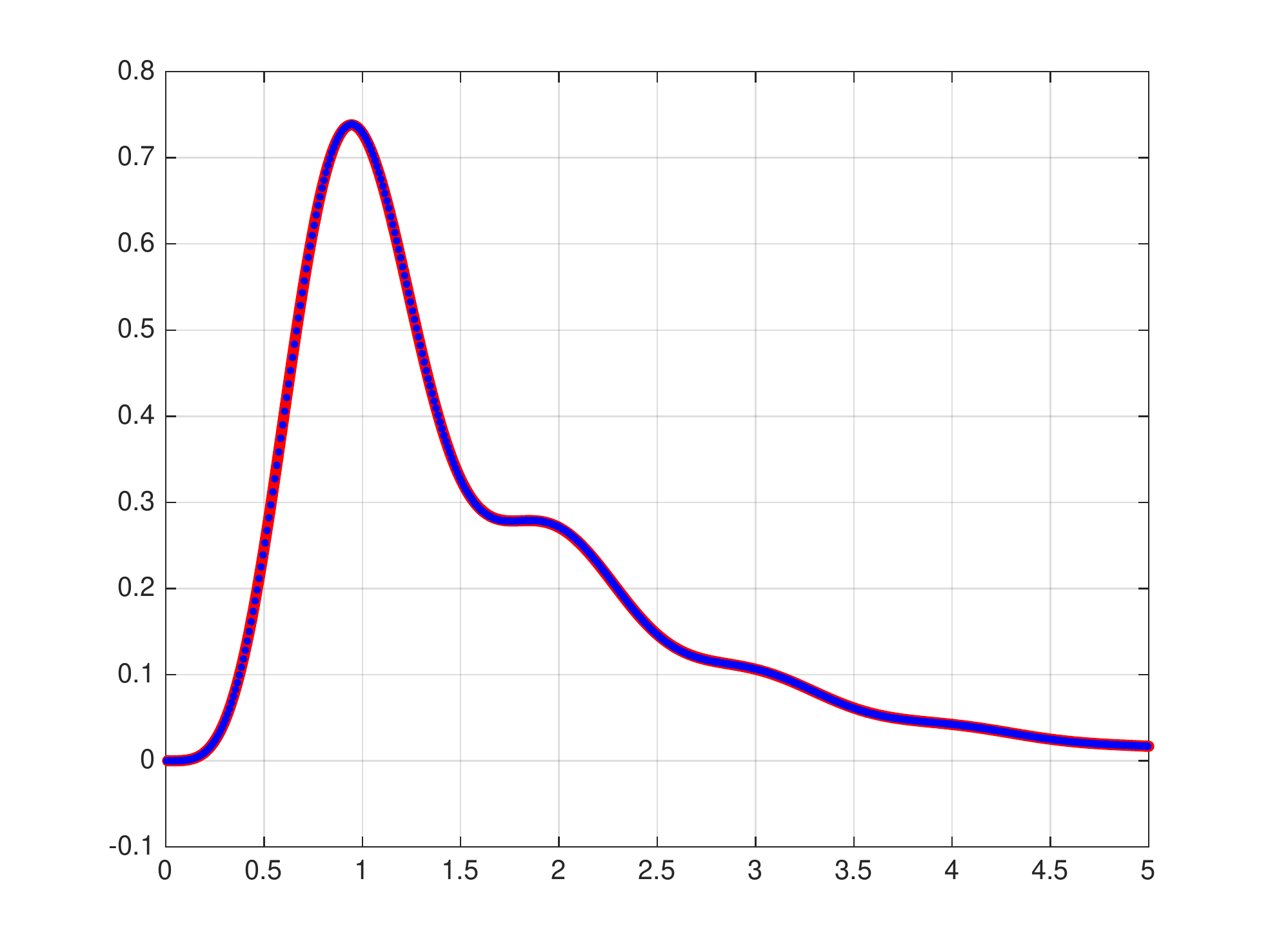}\hspace*{-0.5cm}
\includegraphics[width=0.55\textwidth]{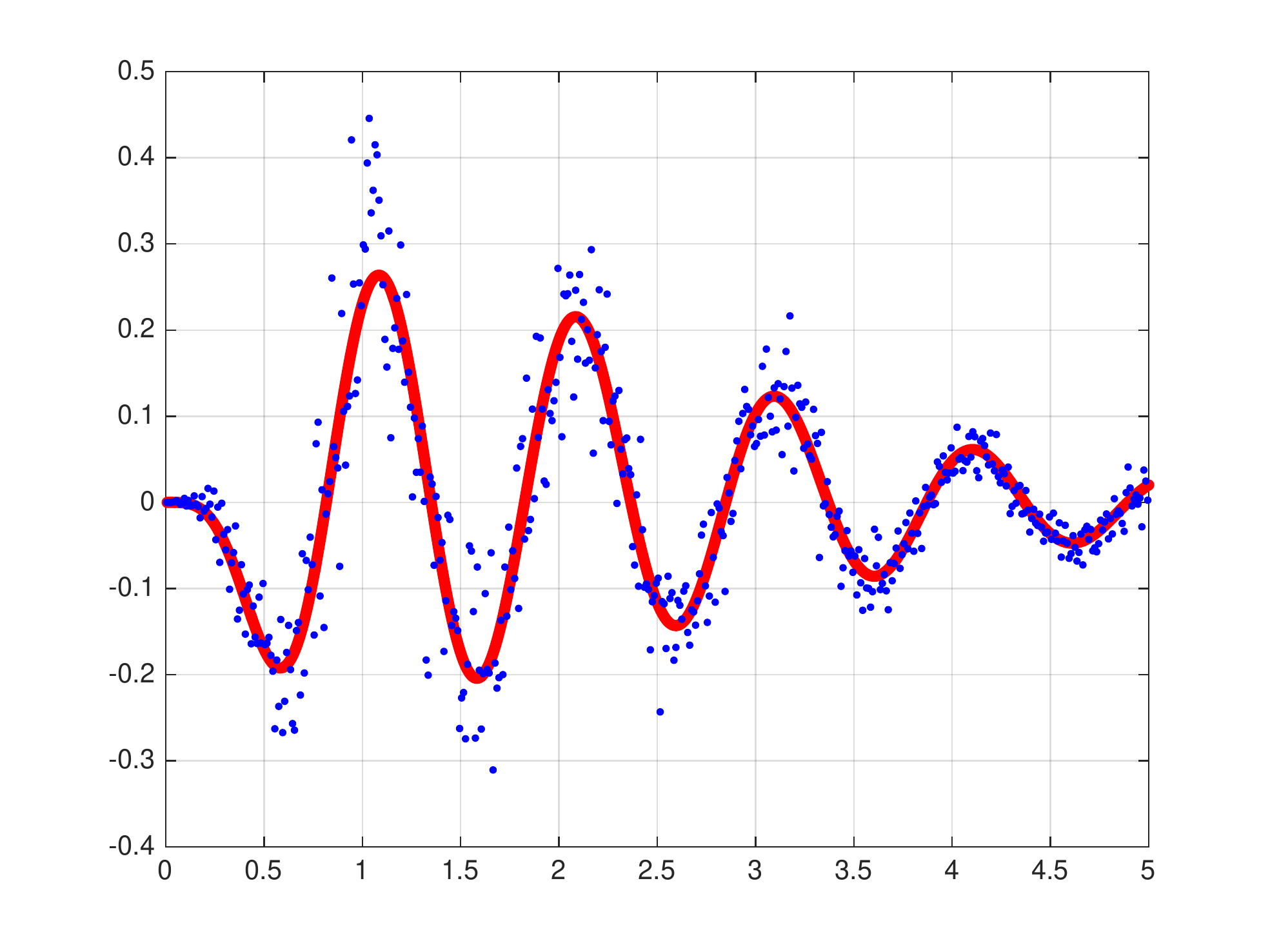}
\caption{\label{F9a}As in Figure \ref{F9} but with thinning $\xi=0.6$.}
\end{figure}

 As was done for the COE, to illustrate and test the above results, we took $10^8$ samples of ${\rm CSE}_N$  ($N=20$, again using the $\beta=4$ CMV sparse matrix model \cite{KN04}), and from this made an empirical determination of
the spacing distribution scaled to unit mean spacing. This was then subtracted from the large $N$ limit
$p_{4,\xi}(0;s)$ and the difference compared against $r_{4,\xi}(0;s)$. Both
$\xi =1$ (no thinning) and $\xi=0.6$ were considered; see Figures \ref{F9} and \ref{F9a}.

\subsection{Painlev\'e transcendent characterisation} 
 
In addition to the integral operator characterisation of the expansion terms in Propositions~\ref{prop:COE}
and \ref{prop:CSE} there is one in terms of Painlev\'e transcendents. Such an expression in the case of the CUE has been given in \cite{FM15} and
restated in \eqref{2.13d}--\eqref{Tq2} above. To simplify we restrict ourselves to $N$ being
even, writing $N \to 2N$ for definiteness.
We begin by noting that
 the joint eigenvalue PDF for the eigen-angles $\{ \theta_j \}$ of $\O^\pm(2N+1)$ matrices in the interval $(0,\pi)$, after changing variables
 $x_j = \sin^2(\theta_j/2)$, is proportional to
 $$
 \prod_{l=1}^N x_l^a (1 - x_l)^b \prod_{1 \le j < k \le N} (x_k - x_j)^2, \qquad 0 < x_l < 1,
 $$
 with $a = \pm 1/2$, $b=-a$, which is an example of the Jacobi unitary ensemble (JUE); see e.g.~\cite[\S 3.7.1]{Fo10}.
 Hence the corresponding generating functions for the gap probabilities are related by
 \begin{equation}\label{Sa1a}
 \mathcal E^{\O^\pm(2N+1)}((0,\theta);z) = \left. \mathcal E^{{\rm JUE}_N}\left ( \left (0,\sin^2{ \theta \over 2} \right );z \right ) \right |_{a = \pm 1/2, b = - a}.
 \end{equation}
 The significance of this result for present purposes is that the RHS can be expressed as a $\sigma$PVI transcendent \cite{HS99,FW04}.
 Specifically, reading off from \cite[Eqns.~(8.71), (8.75), (8.76)]{Fo10} we have
 \begin{equation}\label{Sa2} 
 \mathcal  E^{{\rm JUE}_N}((0,\sin^2 \theta/2);z) = \exp \int_0^{\sin^2 \theta/2} {f^\pm(t;z) \over t ( t - 1)} \, dt,
 \end{equation}
 where $f=f^\pm$ satisfies the particular $\sigma$PVI equation
 \begin{equation}\label{Sa3}  
 (t(1-t)f'')^2 + (f' - N^2)(2 f + (1 - 2t)f')^2 - (f')^2 \Big (f' - N^2 + {1 \over 4} \Big ) = 0,
 \end{equation}
 subject to the boundary condition
 \begin{equation}\label{Sa4} 
 f^\pm(t;z) \mathop{\sim}\limits_{t \to 0^+} \left \{ \begin{array}{ll} \displaystyle {8 N(N^2-1/4) \over 3 \pi} z t^{3/2}, & a=1/2, \\[.2cm]
  \displaystyle {2 N \over \pi} z t^{1/2},  & a=-1/2. \end{array} \right.
 \end{equation}
 We therefore have
 \begin{equation}\label{Sa5}  
  \mathcal E^{\O^\pm(2N+1)}((0,\pi s/N);z) =
 \exp \Big ( - \int_0^{(\pi s)^2} f^\pm(\sin^2 \sqrt{u}/2N;z) \, {du \over N \sqrt{u} \sin \sqrt{u}/N} \Big ).
 \end{equation}

 To obtain an expansion consistent with (\ref{Sa}) we make the ansatz
\begin{equation}\label{ffA}
f^\pm(\sin^2(\sqrt{w}/2N);z) = f_0^\pm + {f^\pm_1 \over N^2} + \O(N^{-4}).
 \end{equation}
 Changing variables in (\ref{Sa3}) $t = \sin^2(\sqrt{w}/ 2N)$, substituting (\ref{ffA}), and equating terms to leading order $(N^4)$ and to next
 leading order $(N^2)$, we obtain characterisations in terms of differential equations of $f_0^\pm$ and $f_1^\pm$.
 
 \begin{prop}
 The leading function $f_0^\pm$ in \eqref{ffA} satisfies the particular {\rm $\sigma{\rm PIII}'$} equation {\rm({\em with $v_1= v_2 = 1/2$ in the notation of \cite[Eq.~(8.15)]{Fo10}}\/)}
\begin{equation}\label{ffA1}
w^2 (f_0'')^2 - {(f_0')^2 \over 4} + 4 f_0 (f_0')^2 + w (f_0')^2 - 4 w (f_0')^3 - f_0 f_0' = 0,
 \end{equation} 
subject to the boundary condition
\begin{equation}\label{ffA2}
f_0^\nu(w;z) \sim \left \{ \begin{array}{ll} \displaystyle  {z \over 3 \pi} w^{3/2}, & \nu = +, \\*[.4cm]
\displaystyle {z \over \pi} w^{1/2}, & \nu = - . \end{array} \right.
\end{equation}
The function $f^\pm_1$ in \eqref{ffA} satisfies the second order linear differential equation
\begin{equation}\label{ffA3}
A_1(w;z) f_1'' + B_1(w;z) f_1' + C_1(w;z) f_1 + D_1(w;z) = 0,
\end{equation}
where the coefficients are given in terms of $f_0 = f_0^\pm$ according to
\begin{align*}
A_1(w;z) & = 2 w^2 f_0'', \\
B_1(w;z) & = 8 f_0 f_0' + 2 w f_0' - 12 w f_0^2 - f_0 - {f_0' \over 2}, \\
C_1(w;z) & = f_0'(4 f_0' - 1), \\
D_1(w;z) & = {f_0' \over 3} \Big ( 3 f_0^2 + w^2 f_0'' - 2 w^2 (f_0')^2 + w f_0 - {w f_0' \over 4} - 2 w f_0 f_0' \Big ) - {f_0^2 \over 4},
\end{align*}
 subject to the boundary condition
 \begin{equation}\label{ffA4}
 f_1^\pm(w;z) \sim - {z w^{3/2} \over 12 \pi},
\end{equation}
which is thus the same in both cases.
In terms of $f_0^\pm$ and $f_1^\pm$ we have 
\begin{multline}\label{3.31}
\mathcal E^{\O^\pm(2N+1)}((0,\pi s/N); z) = \exp \left ( - \int_0^{(\pi s)^2} {f_0^\pm (w;z) \over w} \, dw \right ) \\
\times \left ( 1 - {1 \over N^2} \int_0^{(\pi s)^2} {w f_0^\pm(w;z) + 6 f_1^\pm(w;z) \over 6 w} \, dw + \O\Big ( {1 \over N^4} \Big ) \right ).
\end{multline}
 \end{prop}

\begin{remark}
A problem for future study is to describe how the solutions of \eqref{Sig2} and \eqref{ffA3} relate to the broader Painlev\'e theory.\footnote{This was raised by one of the referees.} 
\end{remark}

Scaling the eigen-angles of ${\rm COE}_{2N}$ to have unit mean spacing by setting $\theta =\pi s/N$, 
we can now substitute (\ref{3.31}) into \eqref{W7}; then---after using \eqref{Th} to add the presence of thinning---substitute the result in (\ref{Sa1}) and compare with (\ref{Sa}) to obtain characterisations of the expansion terms $p_{1,\xi}(0;s)$ and $r_{1,\xi}(0;s)$ in terms of Painlev\'{e} transcendents.
 
 \begin{prop}\label{Cc1}
 We have
 \begin{multline}\label{Cc1a}
 p_{1,\xi}(0;s) = {1 \over \xi (2 - \xi)}  {d^2 \over d s^2} \left ( (1 - \xi) \exp \left ( - \int_0^{(\pi s/2)^2} {f_0^+(w;2 \xi - \xi^2) \over w} \, dw \right ) \right.\\
\left. + \exp  \left ( - \int_0^{(\pi s/2)^2} {f_0^-(w;2 \xi - \xi^2) \over w} \, dw \right ) \right)
 \end{multline}
 and
  \begin{multline}\label{Cc1b}
 r_{1,\xi}(0;s) = \\
 - {4 \over \xi (2 - \xi)} {d^2 \over d s^2} \left (  (1 - \xi)   \int_0^{(\pi s/2)^2} {w f_0^+(w;2\xi - \xi^2) + 6 f_1^+(w;2\xi - \xi^2) \over 6 w} \, dw \right. \\
 \times
 \exp \left ( - \int_0^{(\pi s/2)^2} {f_0^+(w;2\xi - \xi^2) \over w} \, dw \right ) \\
 \quad +  \int_0^{(\pi s/2)^2} {w f_0^-(w;2\xi - \xi^2) + 6 f_1^-(w;2\xi - \xi^2) \over 6 w} \, dw  \\
 \left. \times \exp \left ( - \int_0^{(\pi s/2)^2} {f_0^-(w;2\xi - \xi^2) \over w} \, dw \right ) \right ).
 \end{multline}
 \end{prop}
 
 \begin{remark}\label{R3.6}
 The case $\xi = 1$ of (\ref{Cc1a}) agrees with the $\sigma{\rm PIII}'$
 formula for $p_1(0;s)$ (no thinning) reported in \cite{FW00e} and \cite{Fo06}.
 \end{remark}
 
 The differential equations (\ref{ffA1}) and (\ref{ffA3}) can be used to successively generate terms in the series expansions of $f_0^{\pm}$ and $f_1^{\pm}$ about the origin,
 extending the boundary conditions (\ref{ffA2}) and (\ref{ffA4}). Substituting these into (\ref{Cc1a}) and (\ref{Cc1b}) then provides us with the small $s$
 expansion of $p_{1,\xi}(0;s)$ and $r_{1,\xi}(0;s)$.
 
 \begin{corollary}\label{C3.8}
 We have
  \begin{multline}\label{Cc1c}
  p_{1,\xi}(0;s) = \frac{1}{6} \pi ^2 \xi s   -\frac{1}{60} \pi ^4 \xi s^3  -\frac{1}{270} \pi ^4(\xi
   -2) \xi  s^4 +\frac{\pi ^6 \xi s^5  }{1680}  +\frac{\pi ^6
   (\xi -2) \xi  s^6 }{4725}  \\*[2mm] - \frac{\pi ^8 \xi s^7 }{90720} +\frac{\pi ^8  (\xi
   -2) (3 \xi -32)\xi s^8  }{5292000} +\frac{\pi ^{10} \xi s^9 }{7983360}   + \O(s^{10})
  \end{multline}
  and
 \begin{multline}\label{Cc1d}  
 r_{1,\xi}(0;s) = -\frac{1}{6} \pi ^2 \xi s +\frac{1}{18} \pi ^4 \xi s^3 +\frac{1}{54} 
   \pi ^4 (\xi -2) \xi s^4-\frac{1}{240} \pi ^6  \xi s^5  -\frac{4 \pi ^6
    (\xi -2) \xi s^6}{2025}\\*[2mm] +\frac{\pi ^8 \xi s^7}{7560}-\frac{\pi ^8
    (\xi -2) \left(3 \xi -32 \right)\xi s^8}{352800}-\frac{\pi^{10} \xi s^9}{435456} + \O\left(s^{10}\right).
 \end{multline}    
  \end{corollary}
We remark that setting $\xi=1$ in (\ref{Cc1c}) extends the expansion
of $p_{1,\xi}(0;s)|_{\xi = 1}$ given in \cite[Eq.~(8.141)]{Fo10}.

 For the analogous results in the case of the ${\rm CSE}_N$, scaling the eigen-angles to have unit mean spacing 
 by setting $\theta =2\pi s/N$, we substitute (\ref{3.31}) into 
 (\ref{W6c});  then---after using \eqref{Th} to add the presence of thinning---substitute the result in (\ref{Sa16}) and compare
 with (\ref{Sa7}).

  \begin{prop}\label{Cc2}
 We have
  \begin{multline}\label{Cc2a}
 p_{4,\xi}(0;s) \\
 = {1 \over 2 \xi }  {d^2 \over d s^2} \left (  \exp \left ( - \int_0^{(\pi s)^2} {f_0^+(w; \xi ) \over w} \, dw \right ) +
 \exp  \left ( - \int_0^{(\pi s)^2} {f_0^-(w; \xi ) \over w} \, dw \right ) \right )
 \end{multline}
 and
  \begin{multline}\label{Cc2b}
 r_{4,\xi}(0;s) = - {1 \over 2 \xi } {d^2 \over d s^2} \left (    \left ( \int_0^{(\pi s)^2} {w f_0^+(w;\xi ) + 6 f_1^+(w;\xi) \over 6 w} \, dw \right ) \right.\\ \times
 \exp \left ( - \int_0^{(\pi s)^2} {f_0^+(w;\xi) \over w} \, dw \right ) \\
 \quad \left. + 
 \left ( \int_0^{(\pi s)^2} {w f_0^-(w;\xi ) + 6 f_1^-(w;\xi) \over 6 w} \, dw \right ) 
 \exp \left ( - \int_0^{(\pi s)^2} {f_0^-(w;\xi) \over w} \, dw \right ) 
  \right ).
 \end{multline}
 \end{prop}
 
  \begin{remark}
  As in Remark \ref{R3.6}
 the case $\xi = 1$ of (\ref{Cc2a}) agrees with the $\sigma\text{PIII}'$ 
 formula for $p_4(0;s)$ (no thinning) reported in \cite{FW00e} and \cite{Fo06}.
 \end{remark}
 
 As with Corollary \ref{C3.8} we can use these characterisations to generate power series expansions about the origin.
 
\begin{corollary}\label{C3.9}
 We have
\begin{align}
\nonumber  p_{4,\xi}(0;s) &= \frac{16}{135} \pi ^4 \xi s^4-\frac{128  \pi ^6 \xi  s^6}{4725}+\frac{512 \pi ^8 \xi s^8}{165375}-\frac{34816  \pi ^{10} \xi  s^{10}}{147349125}+ \O(s^{12}),\\*[2mm]
\nonumber r_{4,\xi}(0;s) &= -\frac{4}{27} \pi ^4 \xi s^4+\frac{128  \pi ^6 \xi  s^6}{2025}-\frac{128 \pi ^8 \xi s^8}{11025}+\frac{17408  \pi ^{10} \xi  s^{10}}{13395375}+ \O(s^{12}).
\end{align}
 \end{corollary}

\renewcommand{\thesection}{A}
\section*{Appendix: Implementation}

The method of Section~\ref{sub:numerical} for numerically evaluating general terms of integral operators is most easily added to the Matlab toolbox described
in \cite{Bo09}; the basic code needed to run it is as follows:  

\begin{lstlisting}[numbers=left,basicstyle=\ttfamily\footnotesize,texcl]
function [val,err,n] = OperatorTerm(term,varargin)

%OPERATORTERM evaluates terms of integral operators that have a scalar value.
%
%   OPERATORTERM(term,K1,...,Km) returns the value of the
%   expression 'term(K1,...,Km)' for m discrete integral 
%   operators K1,...,Km (that is, for Nyström matrices of
%   a variable dimension n, which is chosen adaptively). 
%   The value to be returned must be scalar.

tol = 5e-15;
n = 7; n_max = 1000;

operators = cell(size(varargin));

val0 = inf;
while n < n_max
    n = floor(1.41*n);
    for k = 1:length(varargin)
        operators{k} = varargin{k}(n);
    end
    val = term(operators{:});
    err = abs(val-val0);
    if err < tol; break; end
    val0 = val;
end
\end{lstlisting}

From the user perspective this command is called with a considerable symbolic look and feel. For example, the leading  order term in \eqref{Tq0}, that is,
\[
\det(\II-\KK_s)
\]
evaluates for $s=1$ to the following value:

\begin{lstlisting}[numbers=left,basicstyle=\ttfamily\footnotesize,texcl]
s = 1;
K = op(@(x,y) sinc(pi*(x-y)),[0,s]);
I = @(K) eye(size(K));
lead = @(K) det(I(K) - K);
[val,err,n] = OperatorTerm(lead,K);
PrintCorrectDigits(val,err);
 
    0.170217421379185   
\end{lstlisting}
and the leading correction term in \eqref{Tq0}, that is, 
\[
\Om{\KK_s}{\LL_s} = -\det(\II - \KK_s) \tr \left((\II - \KK_s)^{-1}\LL_s \right),
\]
evaluates for $s=1$ to the following value:
\begin{lstlisting}[numbers=left,firstnumber=9,basicstyle=\ttfamily\footnotesize,texcl]
L = op(@(x,y) pi*(x-y).*sin(pi*(x-y))/6,[0,s]);
corr = @(K,L) -det(I(K)-K)*trace((I(K)-K)\L);
[val,err,n] = OperatorTerm(corr,K,L);
PrintCorrectDigits(val,err);
 
    -0.075241982465122
\end{lstlisting}
Since the estimated error has been taken into account when printing these values, they are good to about 15 digits; 
in both cases the adaptively chosen number of Gauss--Legendre quadrature points (that is, the dimension of the 
Nyström matrices representing the
integral operators) was as small as $n=15$, which is a clear sign of the exponential convergence of the method. 
CPU time is about a millisecond.

\section*{Acknowledgements}
We are grateful to A.~Odlyzko for providing us with the Riemann zero data set. We would also like to thank the referees for their thoughtful comments. The research of F.B. was supported by the DFG-Collaborative Research Center, TRR 109, ``Discretization in Geometry and Dynamics.'' The research of P.J.F. and A.M. is part of the program of study supported by the 
ARC Centre of Excellence for Mathematical \& Statistical Frontiers.


\providecommand{\bysame}{\leavevmode\hbox to3em{\hrulefill}\thinspace}
\providecommand{\MR}{\relax\ifhmode\unskip\space\fi MR }
\providecommand{\MRhref}[2]{%
  \href{http://www.ams.org/mathscinet-getitem?mr=#1}{#2}
}
\providecommand{\href}[2]{#2}

\end{document}